\documentclass[journal,onecolumn]{IEEEtran}

\usepackage[english]{babel}
\usepackage{cite}
\usepackage{amssymb,amsthm,mathtools}
\usepackage[ruled,linesnumbered]{algorithm2e}
\usepackage{placeins}
\usepackage{array}
\usepackage[shortlabels]{enumitem}
\usepackage[hidelinks]{hyperref}
\usepackage[capitalize]{cleveref}
\usepackage{color}
\crefname{appendix}{Appendix}{Appendices}
\Crefname{appendix}{Appendix}{Appendices}
\newtheorem{theorem}{Theorem}[section]
\newtheorem{corollary}[theorem]{Corollary}
\newtheorem{lemma}[theorem]{Lemma}
\newtheorem{example}{Example}
\newtheorem{remark}{Remark}[section]
\theoremstyle{definition}
\newtheorem{definition}[theorem]{Definition}

\AddToHook{env/lemma/begin}{\crefalias{theorem}{lemma}}
\AddToHook{env/corollary/begin}{\crefalias{theorem}{corollary}}
\AddToHook{env/definition/begin}{\crefalias{theorem}{definition}}
\AddToHook{env/example/begin}{\crefalias{theorem}{example}}
\AddToHook{env/remark/begin}{\crefalias{theorem}{remark}}

\newcommand{\cA}{\mathcal{A}}
\newcommand{\cC}{\mathcal{C}}
\newcommand{\cP}{\mathcal{P}}

\renewcommand{\le}{\leqslant}
\renewcommand{\leq}{\leqslant}
\renewcommand{\ge}{\geqslant}
\renewcommand{\geq}{\geqslant}

\newcommand{\F}{\mathbb{F}}
\newcommand{\K}{\mathbb{K}}
\newcommand{\N}{\mathbb{N}}
\newcommand{\bbP}{\mathbb{P}}
\newcommand{\vi}{\mathbf{i}}
\newcommand{\vj}{\mathbf{j}}
\newcommand{\vv}{\mathbf{v}}
\newcommand{\vx}{\mathbf{x}}
\newcommand{\vy}{\mathbf{y}}
\newcommand{\vz}{\mathbf{z}}
\newcommand{\va}{\mathbf{a}}
\newcommand{\vb}{\mathbf{b}}
\newcommand{\vc}{\mathbf{c}}
\newcommand{\valpha}{\boldsymbol{\alpha}}

\newcommand{\Zero}{{\mathbf{0}}}
\newcommand{\One}{{\mathbf{1}}}
\newcommand{\eps}{\varepsilon}

\newcommand{\LCS}{\mathrm{LCS}}
\newcommand{\rank}{\mathrm{rank}}
\newcommand{\Set}{\mathrm{Set}}
\newcommand{\Var}{\mathrm{Var}}
\newcommand{\Div}{\mathrm{Div}}
\newcommand{\supp}{\mathrm{Supp}}
\newcommand{\Ev}{\mathrm{Ev}}

\title{Random Algebraic Geometry Codes Approach the Half-Singleton Bound for Insertions and Deletions}
\author{Zhihao~Guan and Hengjia~Wei%
    \thanks{This work was supported in part by the National Natural Science Foundation of China under Grant 12371523.}
	\thanks{Z. Guan and H. Wei are with the School of Mathematics and Statistics, Xi'an Jiaotong University, Xi'an 710049, China (e-mail: guanzh@stu.xjtu.edu.cn; hjwei05@gmail.com).}%
}

\begin{document}
	
	\maketitle
	
	\begin{abstract}
		In this paper, we study the performance of algebraic geometry (AG) codes against adversarial insertion-deletion (insdel) errors. The half-Singleton bound states that an $[n,k]_q$ linear code can correct at most $n-2k+1$ insdel errors. It was recently proven that random Reed-Solomon codes approach this bound. However, these constructions require the field size $q$ to grow linearly with the code length $n$. We overcome this barrier by extending the probabilistic analysis of general linear insdel codes to AG codes. We demonstrate that curves with many rational points allow for nearly optimal codes over significantly smaller alphabets. 
		
		We prove the following main asymptotic results: (1) For general smooth complete curves of fixed genus, random AG codes are nearly optimal, that is, they can correct $(1-\varepsilon)n-2k$ insdel errors with high probability over linear-sized fields ($q=\Theta(n)$). (2) By utilizing Hermitian curves, we achieve this optimality over sublinear fields of size $q=\Theta(n^{2/3})$, breaking the linear field size barrier. (3) Using asymptotically optimal Garc\'{i}a-Stichtenoth towers, we prove the existence of random AG codes that approach the half-Singleton bound with high probability over fields of size $q=2^{O_R(1/\varepsilon^2)}$, independent of $n$.
	\end{abstract}
	
	\section{Introduction}\label{Introduction}
	
	Error-correcting codes are widely used in digital communication and storage systems to ensure data reliability. In classical coding theory, the dominant model is the Hamming metric, where errors alter symbol values but preserve their positions. However, modern applications increasingly face \emph{synchronization errors}, primarily insertions and deletions (insdel errors). Unlike Hamming errors, insertions and deletions alter the sequence length, causing a loss of positional information. This ``desynchronization" means the decoder cannot easily determine which received symbol corresponds to which original index, rendering standard algebraic decoding methods ineffective and making the design of optimal insdel codes significantly more challenging.
	
	The study of codes for these errors began in the 1960s with the foundational work of Levenshtein, Varshamov \cite{lcvenshtcin1966binary} and Tenengolts \cite{tenengolts1984nonbinary}. For decades, constructing efficient codes remained a difficult open problem.  Recently, however, there has been renewed interest in this field, with numerous studies focusing on designing codes capable of correcting insdel errors \cite{guruswami2016efficiently,brakensiek2017efficient,schoeny2017codes,guruswami2017deletion,cheng2018deterministic, haeupler2019optimal, sima2020optimal, guruswami2021explicit,guruswami2022zero,smagloy2023single,yasunaga2024improved,beelen2025reed,sun2025codes}. This resurgence is driven by critical real-world applications. For instance, in DNA-based data storage systems \cite{bornholt2016dna,yazdi2017portable, jain2017duplication,bar2025scalable}, insertion and deletion errors occur naturally during the synthesis and sequencing processes. Similarly, in racetrack memory \cite{chee2017codes}, data is stored in magnetic domains that move along a wire, where imprecise shifts can cause synchronization errors. These technologies necessitate efficient codes capable of handling such disturbances.
	
	Linear codes are particularly desirable for these applications due to their compact representation and support for efficient encoding and decoding algorithms. The study of linear codes for insertion and deletion errors has advanced significantly in recent years. It was first established by Abdel-Ghaffar et al. \cite{abdel2007linear}  that any linear code capable of correcting a single deletion must have a rate of at most $1/2$. Later, Cheng, Guruswami, Haeupler, and Li \cite{cheng2023efficient} established the half-Singleton bound, which implies that any $[n,k]$ linear code can correct at most $n-2k+1$ insdel errors. Ji et al.~\cite{ji2023strict} tightened this limit to the strict half-Singleton bound, showing that codes lacking the all-one vector can correct at most $n-2k$ errors. These bounds serve as the fundamental limits for linear insdel codes, analogous to the Singleton bound for Hamming errors.
	
	Approaching this bound with efficient linear codes has been a major challenge. Cheng, Guruswami, Haeupler, and Li~\cite{cheng2023efficient} constructed linear codes that correct a $\delta$ fraction of insdel errors with rate approaching $(1-\delta)/2$ as the alphabet size increases. Later, Cheng et al.~\cite{cheng2023linear} provided constructions for high-rate and high-noise regimes. A natural candidate for optimal linear codes is the family of Reed-Solomon (RS) codes, which are MDS codes widely used for Hamming errors. However, their performance against insdel errors remained an open problem for decades. Recently, Con, Shpilka, and Tamo~\cite{con2023reed} proved that RS codes can exactly achieve the half-Singleton bound, but their existential argument required a field size exponential in the dimension ($q = O(n^{4k-2})$). Con, Guo, Li, and Zhang~\cite{con2024random} significantly improved this result: they analyzed random RS codes and proved that they approach the half-Singleton bound, that is, correct $(1-\eps)n-2k$ errors with high probability over linear field sizes ($q = \Theta(n)$). In a different direction, Gross, Con, and Yaakobi~\cite{gross2025improved} constructed explicit codes over $\F_{q^2}$, linear over $\F_q$, with rate $1/2-\delta-\eps$ that efficiently correct a $\delta$ fraction of insdel errors for $q=O(\eps^{-4})$; they also gave fully linear codes over $\F_q$ with rate $1/2-2\sqrt{\delta}-\eps$.
	
	While RS codes are arguably the most prominent linear codes, they have a significant limitation: the field size must exceed the code length ($q>n$). This restricts their use over small or fixed alphabets. To overcome this limitation, we turn to algebraic geometry (AG) codes, which generalize Reed-Solomon codes by evaluating functions on algebraic curves at rational points. Curves with many $\F_q$-rational points can support code lengths larger than $q$, making AG codes natural candidates for approaching the half-Singleton bound over smaller, and potentially constant-sized, fields.
	
	In this paper, we investigate the asymptotic performance of AG codes against adversarial insertion and deletion errors. We examine whether the algebraic structure of curves can reduce the required field size while maintaining optimal error correction capability. Specifically, we analyze random codes generated by rational points on irreducible MDS varieties, that is, varieties not contained in any hyperplane passing through the origin, and then apply this framework to AG codes approaching the half-Singleton bound.
	
	\subsection{Our Results}
	In this paper, we extend the probabilistic analysis of insdel codes to the setting of AG codes. Our main contribution is to show that curves with sufficiently many rational points yield random AG codes that approach the half-Singleton bound over smaller alphabets than those required by Reed-Solomon codes. We establish three main theorems corresponding to three regimes of the field size $q$.
	
	Throughout these informal statements, the underlying curve and one-point divisor are fixed. A random AG code is obtained by sampling and ordering $n$ distinct $\F_q$-rational points outside the support of the divisor and using them as evaluation points. Thus, only the evaluation points and their ordering are random.
	
	\begin{theorem}[Informal, see \cref{thm:prob_curve}] 
        For sufficiently large $n$, let $X /\F_q$ be a smooth complete curve of genus independent of $n$, where $q=\Theta(n)$. Then a random $[n,k]$ AG code constructed from $X$, with rate $k/n<(1-\eps)/2$, corrects $(1-\eps)n-2k$ insdel errors with probability tending to $1$ as $n\to \infty$.
	\end{theorem}
	
	This demonstrates that achieving the half-Singleton bound is not unique to Reed-Solomon codes arising from the projective line, but is a general feature of codes generated from MDS varieties.

	\begin{theorem}[Informal, see \cref{thm:prob_Hermitian}] 
        For sufficiently large $n$, let $X$ be a Hermitian curve over $\F_q$, where $q=\Theta(n^{2/3})$. Then a random $[n,k]$ AG code constructed from $X$, with rate $k/n<(1-\eps)/2$, can correct $(1-\eps)n-2k$ insdel errors with probability tending to $1$ as $n\to \infty$.
	\end{theorem}
	
	We focus on Hermitian curves for two reasons. First, Hermitian codes form a well-studied family of AG codes. Second, Hermitian curves attain the Hasse-Weil upper bound and therefore have an abundance of rational points. This allows us to construct codes over the sublinear-sized fields in the theorem above.

	\begin{theorem}[Informal, see \cref{thm:prob_tower}]
		For any fixed rational rate $R\in(0,1/2)$ and any fixed $\eps\in(0,1-2R)$, there exists a prime power $q=2^{O_R(1/\eps^2)}$ and an infinite family of random $[n,k]$ AG codes over $\F_q$ with rate $k/n=R$, such that the codes correct $(1-\eps)n-2k$ insdel errors with probability tending to $1$ as the block length tends to infinity along the family.
	\end{theorem}
	This result overcomes the fundamental limitation of Reed-Solomon codes, proving that structured linear codes can approach the half-Singleton bound even when the alphabet size does not grow with the block length.
	
	\subsection{Our Approach and Related Work}

	Our goal is to provide codes that approach the half-Singleton bound over fields smaller than those required by Reed-Solomon codes.  Our approach consists of three main steps. First, we introduce the notion of HS-optimal varieties and establish a geometric characterization of half-Singleton optimality. Specifically, we prove that an irreducible variety is HS-optimal if and only if it is an MDS variety. This result reduces the study of optimal linear codes for insertion-deletion errors to the verification of certain determinant conditions on algebraic varieties. 

Second, we develop a probabilistic framework for random linear codes whose generator-matrix columns are sampled from rational points of an MDS variety over a finite field. The main challenge is to estimate the probability that the verification matrices associated with the half-Singleton bound fail to have full rank. We overcome this difficulty by combining a rank-certification procedure with geometric estimates on the number of rational points where a generic determinant vanishes. 

Finally, we apply this framework to algebraic geometry codes. We show that generator matrices of suitable AG codes can be represented by sampling points from affine MDS curves. Therefore, by selecting curves with sufficiently many rational points, our probabilistic framework yields AG codes approaching the half-Singleton bound over fields of linear, sublinear, and constant size.

Our geometric approach is inspired by the generic non-vanishing framework developed by Brakensiek, Dhar, and Gopi~\cite{brakensiek2024generalized}. They study polynomial codes generated by points on irreducible varieties and reduce generic determinant conditions on varieties to conditions on formal power series and subsequently to univariate monomials. We adapt this sequence of reductions to the verification matrices arising from insertion-deletion errors. In particular, the determinant conditions and the corresponding monomial base cases in our setting are different from those in polynomial-code constructions, as they are dictated by the half-Singleton bound rather than by the MDS or higher-order MDS conditions.

Our probabilistic analysis builds on the rank-certification framework developed by Con, Guo, Li, and Zhang~\cite{con2024random} for random Reed-Solomon codes. Their method decomposes pairs of subsequences into chains and certifies the resulting blocks sequentially, using a reserve of additional chains to handle failed determinant checks. We adapt this strategy to codes whose generator-matrix columns are sampled from general MDS varieties. The extension is nontrivial: unlike the Reed-Solomon setting, where evaluation points lie on the projective line and failure probabilities can be bounded by counting bad field elements, our setting requires controlling specialization failures on arbitrary algebraic varieties through geometric estimates of rational points.

Our application to algebraic geometry codes is also related to the recent work of Brakensiek, Dhar, Gopi, and Zhang~\cite{brakensiek2024ag}, where AG codes are analyzed through the viewpoint of generator-matrix columns sampled from algebraic varieties. Their work focuses on list-decoding properties and higher-order MDS codes, whereas we study insertion-deletion errors and the verification matrices governing HS-optimality. By combining the geometric characterization, the probabilistic framework, and the abundance of rational points on suitable algebraic curves, we obtain AG codes approaching the half-Singleton bound over substantially smaller alphabets than Reed-Solomon codes.

	\subsection{Organization}
	The paper is organized as follows. In \cref{Preliminaries} we establish necessary conditions for a linear code to correct a prescribed number of insdel errors and recall the algebraic foundations concerning varieties. \cref{Section3} provides the theoretical foundation for our work, where we prove that an irreducible variety is HS-optimal if and only if it is MDS. In other words, generic vectors from a variety generate an MDS code if and only if they generate a code attaining the half-Singleton bound. In \cref{Section4} we develop our core probabilistic framework. Specifically, we use a chain decomposition technique to derive explicit upper bounds on the failure probability of random codes constructed from these varieties.  Finally, in \cref{section5}, we apply this general framework to AG codes.  We recall the necessary background on algebraic curves and derive our main asymptotic results for smooth complete curves of fixed genus, Hermitian curves, and the Garc\'{i}a-Stichtenoth tower, giving fields of linear, sublinear, and constant size, respectively.
	
	\section{Preliminaries}\label{Preliminaries}
	
	In this section, we establish the foundational notation, define the insertion-deletion metric for linear codes, and review the necessary concepts from algebraic geometry used throughout this paper.
	
	Let $\N=\{0,1,\dots\}$ and $\N^+=\{1,2,\dots\}$. 
	For a positive integer $n$, let $[n]=\{1,2,\dots,n\}$. For a prime power $q$, let $\F_q$ denote the finite field with $q$ elements. Let $\F$ be a non-specific field, which can be either finite or infinite. Let $\F^n$ be a vector space over $\F$ with dimension $n$. 
	
	For a vector $\vv \in \F^n$ and a subset of indices $I \subseteq [n]$, we denote by $\vv_I$ the restriction of $\vv$ to the coordinates in $I$. Similarly, for a matrix $G \in \F^{k \times n}$ and $I \subseteq [n]$, let $G_I$ denote the submatrix of $G$ consisting of the columns with indices in $I$. 
	
	Given a matrix $M \in\F^{m\times n}$, we say that $M$ has full row rank if $\rank(M)=m$ and that $M$ has full column rank if $\rank(M)=n$.

	\subsection{Codes and the Half-Singleton Bound}
	
	In this paper, a vector $\vv \in \F^n$ could also be viewed as a string or a sequence over the alphabet $\F$. A \emph{subsequence} of a string $\va$ is a string obtained by removing some (possibly none) of the symbols in $\va$. Let $\va,\vb$ be strings over an alphabet $\F$. A \emph{longest common subsequence} between $\va$ and $\vb$ is a subsequence of both $\va$ and $\vb$ of maximal length. We denote by $\LCS(\va,\vb)$ the length of a longest common subsequence. 
	
	For two vectors $\va, \vb \in \F^n$, the \emph{insertion-deletion (insdel) distance} $d_I(\va, \vb)$ between $\va$ and $\vb$ is the minimal number of insertions and deletions required to transform $\va$ into $\vb$. It can be verified that $d_I(\va,\vb)$ is indeed a metric on $\F^n$.

	It has been proved that the insdel distance between any two vectors can be characterized by their longest common subsequences.
	\begin{lemma}[\cite{do2021explicit}, Lemma 1]
		Let $\va,\vb\in\F^n$. Then we have
		$$d_I(\va,\vb)=2n-2\LCS(\va,\vb).$$
	\end{lemma}

	An $[n,k]$ linear code $\cC\subseteq \F^n$ is a linear subspace of $\F^n$ of dimension $k$. The code rate is defined as the ratio $R=k/n$. 
	The \emph{insdel distance} of the code $\cC \subseteq \F^n$, denoted by $d_I(\cC)$, is the minimum insdel distance between any pair of distinct codewords. A code of distance $d_I(\cC)$ can correct up to $\lfloor (d_I(\cC) - 1)/2 \rfloor$ insdel errors.
	
	The fundamental limit for insdel codes is the half-Singleton Bound. 
	\begin{lemma}[Half-Singleton Bound, cf. Corollary~5.2 in \cite{cheng2023efficient}]\label{lm:half_Singleton}
		Let $n\ge 2k$, and let $\cC$ be an $[n,k]$ linear code over $\F$. Then $\cC$ can correct at most $n-2k+1$ insertion-deletion errors; equivalently,
        $$d_I(\cC) \leq \max\{2(n-2k+2), 2\}.$$
	\end{lemma}
	
	If the code $\cC$ does not contain the all-one vector $\One$, this bound can be tightened to the strict half-Singleton bound:
	\begin{lemma}[Strict Half-Singleton Bound, \cite{ji2023strict}]\label{lm:strict_half-Singleton}
		Let $n\geq 2k+1$. Let $\cC$ be an $[n,k]$ linear code over $\F$. If $\One=(1, 1,\ldots,1)\notin \cC$, then it is capable of correcting at most $n-2k$ insdel errors, that is,
		$$d_I(\cC) \leq \max\{2(n-2k+1), 2\}.$$
	\end{lemma}

	\begin{definition}[Half-Singleton Optimal Codes]
		Let $n\geq 2k$. We say a linear code $\cC$ is \emph{half-Singleton optimal (HS-optimal for abbreviation)} if it achieves the half-Singleton bound, i.e., $\cC$ can correct $n-2k+1$ insdel errors. A linear code $\cC$ is \emph{strict half-Singleton optimal (strict HS-optimal for abbreviation)} if $\One \notin \cC$ and $\cC$ can correct $n-2k$ insdel errors.
	\end{definition}

    These two notions are directly related: shortening an HS-optimal code $\cC$ yields a strictly HS-optimal code.
    
	\begin{lemma}\label{lm:shorten}
		Let $n\geq 2k$. Let $\cC$ be an $[n,k]$ code over $\F$ with $\One\in\cC$. By row transformations we can always suppose that $\cC$ is generated by a matrix of the form
		$$G=\begin{pmatrix}
			\One_{n-1} & 1 \\
			G' & \Zero_{k-1}
		\end{pmatrix}.$$
		Then the $[n-1,k-1]$ code $\cC'$ generated by the matrix $G'$ satisfies $d_I(\cC')\geq d_I(\cC)$. In particular, if $\cC$ is HS-optimal, $\cC'$ is strict HS-optimal. (See \cref{appendix:shorten} for proof).
	\end{lemma}

    In this paper, we show that there exist AG codes over constant-sized fields that approach this bound. To capture this asymptotic behavior, we introduce the following relaxed notion of optimality.
    
	\begin{definition}[Relaxation]\label{def:relaxation}
		Let $\eps>0$. A linear $[n,k]$ code $\cC$ is \emph{$\eps$-relaxed HS-optimal} if it can correct at least
		$$(1-\eps)n-2k+1$$
		insdel errors when $\One\in\cC$, and at least
		$$(1-\eps)n-2k$$
		insdel errors when $\One\notin\cC$.
	\end{definition}

	To analyze the insdel capability of linear codes, we utilize the increasing subsequence framework.
	\begin{definition}[Increasing Subsequence]
		An \emph{increasing subsequence} is a tuple $I=(I_1,\dots,I_\ell)\in[n]^\ell$ such that $I_1<I_2<\cdots<I_\ell$, where $\ell$ is called the \emph{length} of $I$.
	\end{definition}
	
	Note that $I$ can be regarded as a subsequence of the sequence $(1,2,\dots, n)$. For $I=(I_1,\dots,I_\ell)\in [n]^\ell$, we define the set of its indices $Set(I)=\{I_1,\dots,I_\ell\}$. 
	
	For $I,J\in [n]^\ell$, let $I \cap J$ be an increasing subsequence made up of the corresponding equal components of $I$ and $J$, i.e., $I\cap J=(e_1,e_2, \cdots,e_t)$,  where $t \leq \ell$ and $e_i = I_{r_i} = J_{r_i}$ for all $1 \leq i \leq t$. In this case, we say $I$ and $J$ agree on $t$ coordinates.
	
	We now define the specific matrices that certify optimality.
	\begin{definition}\label{def:verification_matrix}
		Let $\cC$ be an $[n,k]$ code generated by matrix $G=(G_1,G_2,\cdots, G_n)$, where $G_i$ is its $i$-th column. For an integer $\ell$ and increasing subsequences $I, J$ of length $\ell$, we define the \emph{verification matrix} $M_{k,\ell,I,J}$ as follows:
		\begin{itemize}
			\item If $\One$ is not a codeword, we define $M_{k,\ell,I,J}$ as the $2k\times \ell$ matrix
			$$M_{k,\ell,I,J}=\begin{pmatrix}
				G_{I_1} & G_{I_2} & \cdots & G_{I_{\ell}} \\
				G_{J_1} & G_{J_2} & \cdots & G_{J_{\ell}}
			\end{pmatrix}.$$
			\item If $\One$ is a codeword, without loss of generality, assume the first row of $G$ is the all-one vector. Let $G'_i$ denote the column obtained from $G_i$ by removing the first entry. We define $M_{k,\ell,I,J}$ as the $(2k-1)\times \ell$ matrix
			$$M_{k,\ell,I,J}=\begin{pmatrix}
				1 & 1 & \cdots & 1\\
				G'_{I_1} & G'_{I_2} & \cdots & G'_{I_\ell} \\
				G'_{J_1} & G'_{J_2} & \cdots & G'_{J_\ell}
			\end{pmatrix}$$
		\end{itemize} 
	\end{definition}
	
	The following lemma generalizes results from \cite{ji2023strict}, providing a sufficient condition for optimality based on the rank of these matrices.
	\begin{lemma}\label{lm:condition}
		Let $\cC$ be an $[n,k]$ linear code, and let $\ell\leq n$ be an integer such that $\ell\geq 2k-1$ when $\One\in\cC$ and $\ell\geq 2k$ when $\One\notin\cC$. Assume that for every pair of increasing subsequences $I,J\in[n]^\ell$ satisfying $\rank(G_{I\cap J})<k$, the corresponding verification matrix $M_{k,\ell,I,J}$ has full row rank. Then
		$$d_I(\cC)\geq 2(n-\ell+1).$$
		
		In particular, we have the following conditions for optimality:
		\begin{itemize}
			\item If $\One \notin \mathcal{C}$, $\ell = 2k$, and $M_{k,2k,I,J}$ has full row rank for all relevant pairs, then $\mathcal{C}$ is strict HS-optimal.
			\item If $\One \in \mathcal{C}$, $\ell = 2k - 1$, and $M_{k,2k-1,I,J}$ has full row rank for all relevant pairs, then $\mathcal{C}$ is HS-optimal.
			\item Let $\eps>0$, and set $\ell=2k-1+\lfloor\eps n\rfloor$ if $\One\in\cC$, and $\ell=2k+\lfloor\eps n\rfloor$ if $\One\notin\cC$. If $\ell\leq n$ and $M_{k,\ell,I,J}$ has full row rank for all relevant pairs, then $\cC$ is $\eps$-relaxed HS-optimal.
		\end{itemize}
	\end{lemma}
	\begin{proof}
		Suppose, toward a contradiction, that two distinct codewords $\va=\vx G$ and $\vb=\vy G$ have a common subsequence of length at least $\ell$. Then there exist increasing subsequences $I,J\in[n]^\ell$ such that $\va_I=\vb_J$. For every index $e\in I\cap J$, we have $(\vx-\vy)G_e=0$. Since $\vx\neq\vy$, it follows that $\rank(G_{I\cap J})<k$.

		First suppose that $\One\notin\cC$. The nonzero row vector $(\vx,-\vy)$ lies in the left kernel of $M_{k,\ell,I,J}$ because
		$$(\vx,-\vy)M_{k,\ell,I,J}=\va_I-\vb_J=\Zero.$$
		Thus $M_{k,\ell,I,J}$ does not have full row rank, contradicting the hypothesis.

		Now suppose that $\One\in\cC$, and write $\vx=(x_1,\vx')$ and $\vy=(y_1,\vy')$ relative to a generator matrix whose first row is $\One$. Then the nonzero row vector $(x_1-y_1,\vx',-\vy')$ lies in the left kernel of the corresponding verification matrix, since
		$$(x_1-y_1,\vx',-\vy')M_{k,\ell,I,J}=\va_I-\vb_J=\Zero.$$
		This again contradicts the full-row-rank hypothesis. Therefore every pair of distinct codewords has longest common subsequence of length at most $\ell-1$, and hence $d_I(\cC)\geq 2(n-\ell+1)$.

		The first two particular cases follow by setting $\ell=2k$ and $\ell=2k-1$, respectively. For the final case, the distance bound shows that the code can correct at least $n-\ell$ insdel errors. Substituting the stated value of $\ell$ in each case gives the two bounds in \cref{def:relaxation}.
	\end{proof}
	
	Sometimes we do not consider $G$ as a fixed matrix, but rather as a matrix of variables, or more specifically, a collection of $n$ column vectors treated as indeterminants. In such cases, we use the notation $M_{k,\ell,I,J}[G]$ to explicitly indicate the dependency of the verification matrix on $G$. Furthermore, for a matrix $M[G]$ dependent on $G$, we denote by $M|_{G_i=\alpha_i \text{ for } i\in S}$ the specialization of $M$ where the $i$-th column $G_i$ is set to the value $\alpha_i$ for all indices $i \in S$, while all other columns remain as variables.

	\subsection{Algebraic Preliminaries}
	
	Let $\F$ be an algebraically closed field, and let $\mathbb{A}^n$ denote the $n$-dimensional affine space over $\F$. For an ideal $I\subseteq\F[x_1,\dots,x_n]$, let $V(I)$ be its common zero set in $\mathbb{A}^n$. An \emph{affine variety} is a subset $X\subseteq\mathbb{A}^n$ of the form $X=V(I)$ for some ideal $I$. Conversely, the polynomials that vanish on $X$ form its \emph{vanishing ideal}
	$$I(X)=\{f\in\F[x_1,\dots,x_n]:f(x)=0\text{ for all }x\in X\}.$$
	
	An ideal $I$ is said to be \emph{radical} if $f^r \in I$ for some integer $r$ implies $f \in I$. By Hilbert's Nullstellensatz, there is a one-to-one correspondence between algebraic varieties and radical ideals; specifically, for any variety $X$, $I(X)$ is a radical ideal and $V(I(X))=X$.
	
	The \emph{Zariski topology} on a variety $X\subseteq\mathbb{A}^n$ is the topology whose closed sets are the subsets $X\cap V(J)$, where $J\subseteq\F[x_1,\dots,x_n]$ is an ideal. A nonempty variety $X$ is called \emph{irreducible} if it cannot be written as the union of two proper Zariski-closed subsets. Equivalently, $X$ is irreducible if and only if $I(X)$ is a prime ideal.

	The \emph{coordinate ring} of $X$ is $\F[X]=\F[x_1,\dots,x_n]/I(X)$, whose elements can be viewed as polynomial functions on $X$. If $X$ is irreducible, then $\F[X]$ is an integral domain; its field of fractions, denoted by $\F(X)$, is the \emph{function field} of $X$. For $f\in\F[X]$, the set
    $$U_f=\{x\in X : f(x) \neq 0\}$$
    is a basic Zariski open set, called a \emph{principal open set}, whose complement is
    $$X\setminus U_f = \{x\in X: f(x)=0\}.$$

    We say that $f$ is \emph{generically non-vanishing} on $X$ if $U_f$ is Zariski dense in $X$, that is, if the closure of $U_f$ is $X$. Thus, $f$ may vanish on a proper closed subset of $X$, but not on a generic point of $X$.

	If $X$ is irreducible, then every nonempty Zariski open subset of $X$ is dense. Consequently, for $f\in\F[X]$, the following statements are equivalent:
	\begin{enumerate}
		\item $f$ is generically non-vanishing on $X$.
		\item There exists a point $x\in X$ such that $f(x)\neq 0$.
		\item $f$ is not the zero element in $\F[X]$.
	\end{enumerate}
	When working over a finite field $\F_q$, we view varieties as being defined over the algebraic closure $\overline{\F}_q$. An $\F_q$-rational point of $X$ is a point on $X$ whose coordinates belong to $\F_q$. The set of all such points is denoted by $X(\F_q)$. 
	
	We conclude with some lemmas summarizing the geometric properties required for our probabilistic analysis.

	\begin{lemma}[\cite{weil1946foundations}]
		If $X_1 \subseteq \F^{k_1}$ and $X_2 \subseteq \F^{k_2}$ are two varieties, then $X_1 \times X_2 \subseteq \F^{k_1+k_2}$ is also a variety. Furthermore, if $X_1$ and $X_2$ are irreducible, then so is $X_1 \times X_2$.
	\end{lemma}

	To locally parameterize these varieties, we utilize formal power series. A formal power series ring over a field $\mathbb{F}$ in variables $z_1, \dots, z_d$, denoted by $\mathbb{F}[[z_1, \dots, z_d]]$, is a generalization of a polynomial ring that allows for infinite sums of monomials.
	\begin{lemma}[Lemma 6 in \cite{brakensiek2024generalized}]\label{lm:parameterize_variety}
		Let $\mathbb{F}$ be an algebraically closed field and $X \subseteq \F^k$ be an irreducible variety of dimension $d$. There exist functions $f_1, \dots, f_k$ (depending on $X$) in the formal power series ring $\mathbb{F}[[z_1, \dots, z_d]]$ such that for any polynomial $g \in \F[x_1,\ldots,x_k]$ over $\F^k$ it is generically non-vanishing over $X$ if and only if $g(f_1,\ldots,f_k) \in \F[[z_1,\ldots,z_d]]$ is non-zero.
		
		Consequently, for the product variety $X^n \subseteq \F^{kn}$, a multivariate polynomial $g$ over $\mathbb{F}^{kn}$ is generically non-vanishing over $X^n$ if and only if the substituted series $g(f_1(\vz_1),\ldots,f_k(\vz_1),\ldots,f_1(\vz_n),\ldots,f_k(\vz_n))$ is non-zero, where $z_i = (z_{i,1}, \dots, z_{i,d})$ are distinct $d$-tuples of formal variables for each $i \in [n]$.
	\end{lemma}
	
	\begin{lemma}[Adapted from Lemma 7 in \cite{brakensiek2024ag}]\label{lm:poly_degenerate}
		Let $X \subseteq \F^{n_1}$ and $Y \subseteq \F^{n_2}$ be irreducible varieties. Let $f(\vx, \vy)$ be a polynomial in variables $\vx = (x_1, \dots, x_{n_1})$ and $\vy = (y_1, \dots, y_{n_2})$ that is generically non-vanishing over the product variety $X \times Y$. Then, there exists a polynomial $g(\vx)$ in the variables $\vx$ such that:
		\begin{enumerate}
			\item $g$ is generically non-vanishing over $X$.
			\item The degree of $g$ satisfies $\deg(g) \leq \deg_{\vx}(f)$, where $\deg_{\vx}(f)$ denotes the partial degree of $f$ with respect to the variables $\vx$.
			\item For any point $\bar{\vx} \in X$ such that $g(\bar{\vx}) \neq 0$, the specialized polynomial $f(\bar{\vx}, \vy)$ is generically non-vanishing over $Y$.
		\end{enumerate}
		(The proof is provided in \cref{appendix:poly_degenerate}).
	\end{lemma}
	\section{MDS Varieties and Half-Singleton Optimality}\label{Section3}
	
	In this section, we prove that an irreducible variety is HS-optimal if and only if it is MDS. In other words, generic points on such a variety generate a code that is optimal with respect to both the Singleton bound for substitution errors and the half-Singleton bound for insdel errors.

	\subsection{Half-Singleton Optimality on Varieties}
	We begin by formally defining the MDS property for varieties.
	
	\begin{definition}[MDS variety]
		Let $\F$ be an algebraically closed field, and let $X\subseteq \F^k$ be an irreducible variety. We say that $X$ is an \emph{MDS variety} if, for every $n\geq k$, a generic $k\times n$ matrix $G$ with columns in $X$ generates an $[n,k]$ MDS code.
	\end{definition}
	\begin{remark}
		The statement ``$G$ generically has a property" means that there exists a nonempty Zariski open subset $U\subseteq X^n$ such that the matrix $G=(G_1,\ldots,G_n)$ has this property for every $(G_1,\ldots,G_n)\in U$. Since $X$ is irreducible, so is $X^n$, and every nonempty Zariski open subset of $X^n$ is dense. Equivalently, the exceptional tuples for which the property fails are contained in a proper Zariski closed subset of $X^n$. Thus, the columns of $G$ are viewed as independent generic points of $X$.
	\end{remark}
	
	Geometrically, the MDS property can be characterized by hyperplanes passing through the origin: a matrix formed by generic points of the variety has full row rank precisely when the variety is not contained in such a hyperplane.
	
	\begin{lemma}\label{lm:MDS_hyperplane}
		Let $X\subseteq \F^k$ be an irreducible variety. Then the following statements are equivalent:
		\begin{enumerate}[(a)]
			\item \label{prop:1} For every $n\geq k$, a generic $k\times n$ matrix $G$ with columns in $X$ has full row rank.
			\item \label{prop:2} $X$ is MDS.
			\item \label{prop:3} $X$ is not contained in any hyperplane passing through the origin.
		\end{enumerate}
	\end{lemma}
	\begin{proof}
		The equivalence between \cref{prop:2} and \cref{prop:3} is shown by Theorem 12 in \cite{brakensiek2024generalized}. \cref{prop:1} follows from \cref{prop:2} by the definition of MDS varieties.
		
		To show that \cref{prop:1} implies \cref{prop:3}, suppose for the sake of contradiction that $X$ is contained in the hyperplane $\{\vx:\va\cdot\vx=0\}$, where $\Zero\neq\va\in\F^k$ is a non-zero vector. Then for any matrix $G$ with columns in $X$, we have $\va G=\Zero$. This implies the rows of $G$ are linearly dependent, contradicting the fact that $G$ has full row rank.
	\end{proof}
	
	Recall from Section~\ref{Preliminaries} that the presence of the all-one vector $\One$ dictates whether a code can achieve the standard half-Singleton bound or the strict half-Singleton bound. The following lemma translates this algebraic condition into a geometric one, showing that $\One \in \mathcal{C}$ depends entirely on whether the variety $X$ is contained in an affine hyperplane.
	\begin{lemma}\label{lm:code_containing_1}
		Let $X\subseteq \F^k$ be an irreducible MDS variety, and let $n\geq k+1$. For a generic tuple $(G_1,\ldots,G_n)\in X^n$, the code $\cC$ generated by $G=(G_1,\ldots,G_n)$ contains $\One$ if and only if $X$ is contained in an affine hyperplane.
	\end{lemma}
	\begin{proof}
		If $X$ is contained in an affine hyperplane, then by \cref{lm:MDS_hyperplane}, this hyperplane does not pass through the origin. Therefore, it can be written as $\{\vx:\va\cdot\vx=1\}$ for some $\Zero\neq\va\in\F^k$. Then for every tuple $(G_1,\ldots,G_n)\in X^n$, we have
		$$\One = \va G\in \cC.$$ 
		
		Conversely, suppose $X$ is not contained in any affine hyperplane. Thus its affine span is the entire space $\mathbb{F}^k$. Therefore, a choice of $n$ generic points $G_1,G_2,\dots,G_n\in X$ will also have an affine span of dimension $k$. If the code $\mathcal{C}$ contained the all-one vector $\One$, there would exist some $\va\in\mathbb{F}^k$ such that $\va G=(1,\dots,1)$. This implies that $\va\cdot G_i=1$ for all $i\in[n]$, meaning all $n$ points lie on the affine hyperplane defined by $\va\cdot\vx=1$. However, because the affine span of $n$ generic points is $k$, they cannot be contained in a $(k-1)$-dimensional affine hyperplane. Thus, there is no $\va\in\mathbb{F}^k$ such that $\va G=\One$, meaning $\One\notin\mathcal{C}$ generically.
	\end{proof}
	
	We now define HS-optimality for varieties.
	\begin{definition}\label{def:variety_half_bound}
		Let $\F$ be an algebraically closed field, and let $X\subseteq\F^k$ be an irreducible variety. We say that $X$ is \emph{HS-optimal} if, for every integer $n\geq 2k$, a generic matrix $G$ with columns in $X$ generates an $[n,k]$ code $\cC$ that
		\begin{itemize}
			\item contains the all-one codeword $\One$ and can correct $n-2k+1$ insdel errors; or
			\item does not contain the all-one codeword and can correct $n-2k$ insdel errors.
		\end{itemize}
	\end{definition}

	We now state the main result of this section: for irreducible varieties, the MDS property is equivalent to HS-optimality.
    
	\begin{theorem}\label{thm:Variety}
		Let $\mathbb{F}$ be an algebraically closed field, and let $X \subseteq \mathbb{F}^k$ be an irreducible variety. Then $X$ is an MDS variety if and only if it is an HS-optimal variety.
	\end{theorem}
	
	\noindent\textit{Proof overview:}
		The direction from HS-optimality to the MDS property is immediate. If $X$ is not an MDS variety, then by \cref{lm:MDS_hyperplane}, a matrix formed by generic points of $X$ does not have full row rank and therefore does not generate an $[n,k]$ code.

		For the converse, suppose that $X$ is an MDS variety. 
        We need to prove that the verification matrices $M_{k,\ell,I,J}$ are generically nonsingular. The proof proceeds through three steps. First, in \cref{subsection3.2}, \cref{lm:parameterize_variety} translates generic non-vanishing on $X$ into non-vanishing of formal power series. Second, in \cref{subsection3.3}, we reduce the resulting multivariate power series to univariate monomials. Finally, in \cref{subsection3.4}, we prove the required determinant conditions for these monomials. Applying the reductions in reverse proves the result for generic points on $X$.


		%
	
	\subsection{Reduction to Power Series}\label{subsection3.2}
	
	In this subsection, we formalize the first step of our reduction by translating the geometric properties of varieties into algebraic conditions on formal power series. Recall that a formal power series ring over a field $\mathbb{F}$ in variables $z_1, \dots, z_d$, denoted by $\mathbb{F}[[z_1, \dots, z_d]]$, is a generalization of a polynomial ring that allows for infinite sums of monomials. Rather than evaluating these series at specific numerical points, we treat them as formal algebraic objects where properties like linear independence are determined by their coefficients.
	
	\begin{definition}
		Let $\F$ be an algebraically closed field. Given a $k$-tuple of power series $F=(f_1,\ldots,f_k)\in (\F[[z_1,\ldots,z_d]])^k$, we define the $k \times n$ matrix $G^F$ as:
		\begin{align*}
			G^F&=(F(\vz_1),\ldots,F(\vz_n))\\
			&=\begin{pmatrix}
				f_1(\vz_1) & f_1(\vz_2) & \cdots & f_1(\vz_n) \\
				f_2(\vz_1) & f_2(\vz_2) & \cdots & f_2(\vz_n) \\
				\vdots & \vdots & \ddots &\vdots \\
				f_k(\vz_1) & f_k(\vz_2) & \cdots & f_k(\vz_n)
			\end{pmatrix}.
		\end{align*}
		where $\vz_i=(z_{i,1},\ldots,z_{i,d})$, $i\in [n]$, are distinct $d$-tuples of variables. Each entry in $G^F$ is an element in the formal power series ring.
	\end{definition}
	
	\begin{definition}[MDS power series]\label{def:MDS_power_series}
		A tuple of power series $F=(f_1,\ldots,f_k)\in (\F[[z_1,\ldots,z_d]])^k$ is said to be MDS if the following determinant is nonzero:
		$$\det\begin{pmatrix}
			f_1(\vz_1) & f_1(\vz_2) & \cdots & f_1(\vz_k) \\
			f_2(\vz_1) & f_2(\vz_2) & \cdots & f_2(\vz_k) \\
			\vdots & \vdots & \ddots &\vdots \\
			f_k(\vz_1) & f_k(\vz_2) & \cdots & f_k(\vz_k) \\
		\end{pmatrix}\neq 0.$$
	\end{definition}
	
	\begin{lemma}[Lemma 8 in \cite{brakensiek2024generalized}]
		A $k$-tuple $F=(f_1,f_2,\ldots,f_k)\in (\F[[z_1,\ldots,z_d]])^k$ is MDS if and only if $f_1,\ldots,f_k$ are linearly independent as power series over $\F$.
	\end{lemma}
	
	\begin{definition}[HS-optimal power series]\label{def:optimal_power_series}
		Let $F=(f_1,\ldots,f_k)\in (\F[[z_1,\ldots,z_d]])^k$. We say that $F$ is \emph{HS-optimal} if, for every $n\geq 2k$, the following condition holds:
		\begin{itemize}
			\item If $1\in\operatorname{Span}_{\F}\{f_1,\ldots,f_k\}$, apply an invertible row transformation and assume that $f_1=1$. Write $F'=(f_2,\ldots,f_k)^T$. Then, for all increasing subsequences $I,J\in[n]^{2k-1}$ satisfying $|I\cap J|<k$,
			\begin{equation*}
				\det\begin{pmatrix}
					1 & 1 & \cdots & 1\\
					F'(\vz_{I_1}) & F'(\vz_{I_2}) & \cdots & F'(\vz_{I_{2k-1}}) \\
					F'(\vz_{J_1}) & F'(\vz_{J_2}) & \cdots & F'(\vz_{J_{2k-1}})
				\end{pmatrix}\neq 0.
			\end{equation*}
			\item If $1\notin\operatorname{Span}_{\F}\{f_1,\ldots,f_k\}$, then, for all increasing subsequences $I,J\in[n]^{2k}$ satisfying $|I\cap J|<k$,
			\begin{equation*}
				\det\begin{pmatrix}
					F(\vz_{I_1}) & F(\vz_{I_2}) & \cdots & F(\vz_{I_{2k}}) \\
					F(\vz_{J_1}) & F(\vz_{J_2}) & \cdots & F(\vz_{J_{2k}})
				\end{pmatrix}\neq 0.
			\end{equation*}
		\end{itemize}
	\end{definition}
	
	Every HS-optimal tuple is MDS. Indeed, if $f_1,\ldots,f_k$ are linearly dependent, then the rows of each corresponding verification matrix would be linearly dependent, contradicting its nonzero determinant.
	
	\begin{lemma}\label{thm:Power_series}
		Every MDS tuple $F=(f_1,\ldots,f_k)\in(\F[[z_1,\ldots,z_d]])^k$ is HS-optimal.
	\end{lemma}
	
	\noindent\textit{The first reduction:}
	Assume \cref{thm:Power_series}. Let $X\subseteq\F^k$ be an irreducible MDS variety, and let $F=(f_1,\ldots,f_k)$ be the tuple given by \cref{lm:parameterize_variety}. The polynomial $g(\vx_1,\ldots,\vx_k)=\det(\vx_1,\ldots,\vx_k)$ is generically nonzero on $X^k$. By \cref{lm:parameterize_variety}, $g(F(\vz_1),\ldots,F(\vz_k))$ is nonzero, and hence $F$ is MDS.

	We next distinguish the two cases in \cref{def:optimal_power_series}. If $X$ is contained in an affine hyperplane, then by \cref{lm:MDS_hyperplane}, it is contained in a hyperplane $\{\vx:\va\cdot\vx=1\}$. The polynomial $\va\cdot\vx-1$ vanishes on $X$, so \cref{lm:parameterize_variety} gives $\va\cdot F=1$. Thus, $1\in\operatorname{Span}_{\F}\{f_1,\ldots,f_k\}$. Conversely, if $\va\cdot F=1$ for some $\va\in\F^k$, then the same lemma implies that $\va\cdot\vx-1$ vanishes on $X$, and hence $X$ is contained in an affine hyperplane.

	By \cref{thm:Power_series}, the tuple $F$ is HS-optimal. Applying \cref{lm:parameterize_variety} to each verification determinant shows that it is generically nonzero on $X^n$. Therefore, $X$ is HS-optimal, which proves the nontrivial direction of \cref{thm:Variety}.

	\subsection{Reduction of Power Series to Univariate Monomials}\label{subsection3.3}
	
	In the previous subsection, we established that proving \cref{thm:Power_series} is sufficient to prove our main result, \cref{thm:Variety}. In this subsection, we reduce \cref{thm:Power_series} to a base case involving univariate monomials. We first use an auxiliary variable to retain suitable leading monomials of the power series. Since only finitely many monomials remain, we then encode the multivariate monomials as univariate ones.
	
	For a multi-index $\vj\in\N^d$, we denote the multivariate monomial by $\vz^\vj=\prod_{m=1}^d z_m^{j_m}$. By treating the $k$-tuple $F=(f_1,f_2,\ldots,f_k)\in (\F[[z_1,\ldots,z_d]])^k$ as a column vector of power series, we can expand it as
	\begin{equation*}
		\begin{aligned}
			F(\vz)=
			\begin{pmatrix}
				f_1(\vz) \\ f_2(\vz) \\ \vdots \\ f_k(\vz)
			\end{pmatrix}
			=\begin{pmatrix}
				\sum_{\vj \in \N^d} a_{1,\vj}\vz^{\vj} \\ \sum_{\vj \in \N^d} a_{2,\vj}\vz^{\vj} \\ \vdots \\ \sum_{\vj \in \N^d} a_{k,\vj}\vz^{\vj}
			\end{pmatrix}
			=\sum_{\vj \in \N^d} \va_{\vj}\vz^{\vj},
		\end{aligned}
	\end{equation*}
	where $\va_\vj$ is a coefficient column vector whose $i$-th entry is the coefficient of $\vz^{\vj}$ in $f_i$.
    
    Fix a graded lexicographic order on $\N^d$ and use the same order for the corresponding monomials $\vz^\vi$, $\vi\in\N^d$.
	Under this ordering, if $F$ is an MDS tuple, we can identify its linearly independent leading terms.
	
	\begin{lemma}[Lemma 7 in \cite{brakensiek2024generalized}]\label{lm:coefficient_choose}
		Let $F=(f_1,f_2,\ldots,f_k)\in (\F[[z_1,\ldots,z_d]])^k$ be a linearly independent tuple, and write $F(\vz)=\sum_{\vj\in\N^d}\va_\vj\vz^\vj$. There exist multi-indices $\vj_1^*<\cdots<\vj_k^*$ such that the coefficient vectors $\va_{\vj_1^*},\ldots,\va_{\vj_k^*}$ are linearly independent and, for every $\beta\in[k]$ and every $\vj<\vj_\beta^*$,
		$$\va_\vj\in\operatorname{Span}_{\F}\{\va_{\vj_1^*},\ldots,\va_{\vj_{\beta-1}^*}\},$$
		where the span is $\{\Zero\}$ when $\beta=1$.
	\end{lemma}
	After an invertible row transformation, we may therefore assume that $\va_{\vj_1^*},\ldots,\va_{\vj_k^*}$ are the standard basis vectors. It follows that $\vz^{\vj_i^*}$ is the leading monomial of the $i$-th transformed series for every $i\in[k]$.
	
	This allows us to reduce the problem to monomials. The strict case requires a small additional argument because the constant terms need not vanish even though the constant series does not lie in the span of the tuple.
	\begin{lemma}\label{lm:reduction_monomial}
		Suppose that the following tuples are HS-optimal in their respective cases:
		\begin{itemize}
			\item $(1,t^{e_1},\ldots,t^{e_{k-1}})$ for all distinct positive integers $e_1,\ldots,e_{k-1}$;
			\item $(t^{e_1},\ldots,t^{e_k})$ for all distinct positive integers $e_1,\ldots,e_k$.
		\end{itemize}
		Then \cref{thm:Power_series} holds.
	\end{lemma}
		\begin{proof}
		Let $F$ be an MDS tuple. Then the components $f_1,\ldots,f_k$ are linearly independent. We apply \cref{lm:coefficient_choose} to identify the leading multi-indices $\vj_1^*<\cdots<\vj_k^*$. As discussed above, we may apply an invertible change of basis to the span of $F$ so that the corresponding leading coefficient vectors become the standard basis vectors. Following this transformation, the leading term of the $i$-th component $f_i(\vz)$ is exactly the monomial $\vz^{\vj_i^*}$.

		We now employ a weighted substitution to map these multivariate power series to univariate ones while preserving their leading terms. For a positive integer vector $\mathbf{w}=(w_1,\ldots,w_d)$, define the weight of $\vz^{\vj}$ by
		$$|\vj|_{\mathbf{w}}=\sum_{m=1}^d w_mj_m.$$
		Write $|\vj|=\sum_mj_m$ and let $D_0=\max_i|\vj_i^*|$. Since the graded lexicographic order first compares total degrees and only finitely many monomials have degree at most $D_0$, we may choose $\mathbf{w}$ so that every $\vz^{\vj_i^*}$ has smaller weight than all the later monomials appearing in $f_i$. Indeed, first choose weights that preserve the lexicographic order among monomials of the same degree at most $D_0$, and then add the same sufficiently large integer to all the weights. This leaves comparisons within the same degree unchanged and ensures that every monomial whose degree exceeds $|\vj_i^*|$ has weight larger than $|\vj_i^*|_{\mathbf{w}}$, for every $i$. For each $\vz_u$, where $u\in[n]$, substitute
		$$z_{u,m}\mapsto t_u^{w_m}\quad\text{for }m\in[d].$$
		Under this map, $\vz_u^{\vj}$ becomes $t_u^{|\vj|_{\mathbf{w}}}$. Let $e_i=|\vj_i^*|_{\mathbf{w}}$. The weights are chosen so that $e_1<\cdots<e_k$. If $\widetilde F(t_u)$ denotes the tuple obtained from $F(\vz_u)$ by this substitution, then its $i$-th component has the form
		$$\widetilde f_i(t_u)=t_u^{e_i}+\sum_{e>e_i}c_{i,e}t_u^e.$$

		Suppose first that $1\in\operatorname{Span}_{\F}\{f_1,\ldots,f_k\}$. Choose $f_1=1$ and let the other basis elements have zero constant term. Applying the preceding leading-term construction to this basis gives $e_1=0$. For fixed increasing subsequences $I,J\in[n]^{2k-1}$ with $|I\cap J|<k$, let $D(\vz_1,\ldots,\vz_n)$ be the corresponding verification determinant:
		$$D=\det\begin{pmatrix}
			1&\cdots&1\\
			f_2(\vz_{I_1})&\cdots&f_2(\vz_{I_{2k-1}})\\
			\vdots&\ddots&\vdots\\
			f_k(\vz_{I_1})&\cdots&f_k(\vz_{I_{2k-1}})\\
			f_2(\vz_{J_1})&\cdots&f_2(\vz_{J_{2k-1}})\\
			\vdots&\ddots&\vdots\\
			f_k(\vz_{J_1})&\cdots&f_k(\vz_{J_{2k-1}})
		\end{pmatrix}.$$
		Let $\widetilde D(t_1,\ldots,t_n)$ be the determinant obtained from $D$ by the above substitution. To isolate the leading monomials, replace every $t_u$ with $\beta t_u$, where $\beta$ is a formal variable. Factoring out $\beta^{e_i}$ from the two rows associated with $\widetilde f_i$, for every $2\leq i\leq k$, gives
		\begin{align*}
			\widetilde D(\beta t_1,\ldots,\beta t_n)
			={}&\beta^{2\sum_{i=2}^ke_i}
			\det\begin{pmatrix}
				1&\cdots&1\\
				t_{I_1}^{e_2}&\cdots&t_{I_{2k-1}}^{e_2}\\
				\vdots&\ddots&\vdots\\
				t_{I_1}^{e_k}&\cdots&t_{I_{2k-1}}^{e_k}\\
				t_{J_1}^{e_2}&\cdots&t_{J_{2k-1}}^{e_2}\\
				\vdots&\ddots&\vdots\\
				t_{J_1}^{e_k}&\cdots&t_{J_{2k-1}}^{e_k}
			\end{pmatrix}
			+\text{terms of higher degree in }\beta.
		\end{align*}
		By the first assumption, the leading determinant is nonzero. Thus, $\widetilde D(\beta t_1,\ldots,\beta t_n)$ is nonzero. Since it is obtained from $D$ by specialization, the original determinant $D$ is also nonzero.

		If $1\notin\operatorname{Span}_{\F}\{f_1,\ldots,f_k\}$, the same argument applies to the strict verification determinant. When all constant terms vanish, its lowest-degree coefficient in $\beta$ is the strict verification determinant for $(t^{e_1},\ldots,t^{e_k})$, which is nonzero by the second assumption. Otherwise, choose a leading-term basis in which one element has constant term $1$ and the others have zero constant term, and let $p>0$ be the smallest exponent of a nonconstant term in that element after substitution. Subtract its upper row from the corresponding lower row. After interchanging $I$ and $J$ if necessary, choose $J_s\notin\Set(I)$; after the same substitution and scaling, the coefficient of $t_{J_s}^{p}$ in the resulting leading determinant is a standard verification determinant for $(1,t^{e_1},\ldots,t^{e_{k-1}})$, where $e_1,\ldots,e_{k-1}$ are the positive leading exponents of the remaining basis elements. It is therefore nonzero by the first assumption. Hence the strict verification determinant is nonzero in both cases, completing the proof.
	\end{proof}

	\subsection{Monomial Codes}\label{subsection3.4}
	We now verify the base case of our reduction: codes generated by univariate monomials. This family includes the generalized Reed-Solomon codes.
	
	\begin{definition}[Generic Reed-Solomon Code, \cite{brakensiek2023generic}]
		Let $\F$ be an arbitrary field. For a positive integer $n\in \N$, let $\alpha_1,\dots,\alpha_n$ be algebraically independent variables, and let $\K$ be the field of fractions of the polynomial ring $\F[\alpha_1, \ldots, \alpha_n]$. A \emph{generic $[n,k]$ Reed-Solomon code} is the subspace $\cC_{gen} \subseteq \K^n$ generated by the matrix
		\begin{equation}\label{Matrix:Generic_Reed-Solomon}
			\begin{pmatrix}
				1 & 1 & \cdots & 1 \\
				\alpha_1 & \alpha_2 & \cdots & \alpha_n \\
				\vdots & \vdots & \ddots & \vdots \\
				\alpha_1^{k-1} & \alpha_2^{k-1} & \cdots & \alpha_n^{k-1}
			\end{pmatrix}
		\end{equation}
	\end{definition}
	
	\begin{remark}
		The variables $\alpha_1, \dots, \alpha_n$ represent generic evaluation points. By substituting these variables with specific, distinct elements from a finite field, one recovers a standard Reed-Solomon code. 
	\end{remark}
	
	In their foundational work, Con, Shpilka, and Tamo \cite{con2023reed} analyzed the verification matrices of generic Reed-Solomon codes. Specifically, as detailed in the extended preprint of their paper, they demonstrated a crucial algebraic property: if a code is generated by the matrix \eqref{Matrix:Generic_Reed-Solomon} and the elements $\alpha_1,\dots,\alpha_n$ are treated as formal variables, then for any two increasing subsequences agreeing on at most $k-1$ coordinates, the determinant of the corresponding verification matrix is nonzero as a polynomial.
	By combining this non-vanishing determinant property with the rank condition established in \cref{lm:condition}, we immediately obtain the following consequence.
	
	\begin{lemma}\label{lm:generic_RS} 
		A generic $[n,k]$ Reed-Solomon code is HS-optimal.
	\end{lemma}
	
	We now extend this result to monomial codes, where the generators are arbitrary distinct monomials $t^{e_1}, \dots, t^{e_k}$ rather than consecutive powers $1, t, \dots, t^{k-1}$. This generalization is necessary because the reduction in \cref{subsection3.3} produces non-consecutive exponents.

	\begin{lemma}\label{thm:monomial}
		Let $n\geq 2k$ and let $0\leq e_1<e_2<\cdots<e_k$ be integers. Let $\alpha_1, \dots, \alpha_n$ be generic points. Consider the code generated by the matrix
		$$G=\begin{pmatrix}
			\alpha_1^{e_1} & \alpha_2^{e_1} & \cdots & \alpha_n^{e_1} \\
			\alpha_1^{e_2} & \alpha_2^{e_2} & \cdots & \alpha_n^{e_2} \\
			\vdots & \vdots & \ddots & \vdots \\
			\alpha_1^{e_k} & \alpha_2^{e_k} & \cdots & \alpha_n^{e_k}
		\end{pmatrix}.$$
		If $e_1=0$, this code is HS-optimal. If $e_1>0$, it is strict HS-optimal.
	\end{lemma}
	
		\begin{proof}
		We first prove the standard case $e_1=0$. Since the $\alpha_i$ are algebraically independent, every generalized Vandermonde determinant formed by $k$ columns of $G$ is nonzero. Thus, $\rank(G_{I\cap J})<k$ if and only if $|I\cap J|<k$. By \cref{lm:condition}, it suffices to prove that the standard verification determinant $\det(M_{k,2k-1,I,J})$ is nonzero whenever $I,J\in[n]^{2k-1}$ satisfy $|I\cap J|<k$.

		We proceed by induction on $k$. For $k=1$, the verification matrix is $(1)$, so its determinant is nonzero. Now suppose $k\geq2$. If $|I\cap J|=k-1$, choose a position $i$ with $I_i=J_i$ and a position $j$ such that $J_j\notin\Set(I)$. If $|I\cap J|<k-1$, choose different positions $i,j$ such that $I_i\notin\Set(J)$ and $J_j\notin\Set(I)$. Delete positions $i$ and $j$ from both sequences to obtain $I',J'\in[n]^{2k-3}$. By construction, $|I'\cap J'|<k-1$.

		Consider the coefficient of $\alpha_{I_i}^{e_k}\alpha_{J_j}^{e_k}$ in $\det(M_{k,2k-1,I,J})$. These factors must be selected from the two rows corresponding to $e_k$ and columns $i,j$, respectively. Hence this coefficient, up to sign, is
		$$\det(M_{k-1,2k-3,I',J'}),$$
		which is nonzero by the induction hypothesis. Therefore, $\det(M_{k,2k-1,I,J})\neq0$, and the code is HS-optimal.

		The case $e_1>0$ follows similarly, using the strict verification matrix.
	\end{proof}

    \begin{proof}[Proof of \cref{thm:Variety}]
		Suppose first that $X$ is HS-optimal. By definition, a matrix formed by generic points of $X$ has full row rank. Hence $X$ is not contained in a hyperplane passing through the origin, and \cref{lm:MDS_hyperplane} implies that $X$ is MDS.

		Conversely, suppose that $X$ is MDS. The zero-dimensional case is immediate, since an MDS point can occur only when $k=1$. We may therefore assume that $X$ has positive dimension. By \cref{thm:monomial}, univariate monomial tuples satisfy the required verification determinant conditions in both cases: the standard case when $e_1=0$ and the strict case when $e_1>0$. Applying \cref{lm:reduction_monomial} proves \cref{thm:Power_series}. Finally, \cref{lm:parameterize_variety} transfers the non-vanishing of these power series determinants back to generic points on $X$. Therefore, $X$ is HS-optimal.
	\end{proof}

	\section{Probabilistic Analysis of Random Codes on MDS Varieties}\label{Section4}
	
	In \cref{Section3}, we established that codes generated by generic points on an irreducible MDS variety that is contained in an affine hyperplane are half-Singleton optimal. However, ``generic points" inherently assume an infinite field, which is not realizable in practical storage or communication systems. To construct practical codes, we must restrict our attention to \emph{rational points} over a finite field $\F_q$.
	
	Our goal in this section is to prove that this optimal performance is asymptotically preserved with high probability, when the columns of the generator matrix are sampled uniformly from the rational points $S \subseteq X(\F_q)$. We first develop a general probabilistic framework and state the main theorem, and then provide detailed estimates for the relevant failure probabilities.
    In the standard case, when $X$ is contained in an affine hyperplane, we may normalize this hyperplane to $\{x_1=1\}$ by an invertible change of coordinates, which preserves all rank conditions. In the strict case, $X$ is not contained in any affine hyperplane, so no such normalization is made.

	\subsection{Probabilistic Analysis of Random Codes}\label{subsection:probabilistic_analysis}
	
	Let $X \subseteq \overline{\F}_q^k$ be an irreducible MDS variety contained in an affine hyperplane, and let $S\subseteq X(\F_q)$ be a set of $\F_q$-rational points on $X$. Let $G$ be a $k \times n$ matrix whose columns are sampled uniformly and without replacement from $S$.

    Fix $\eps > 0$.  We are going to show that with high probability, the code generated by $G$ is $\eps$-relaxed HS-optimal. Denote 
    $$h=2k-1 \quad \textrm{ and } \quad \ell' = h + \lfloor \eps n \rfloor.$$
	By \cref{lm:condition}, the code fails to be $\eps$-relaxed HS-optimal if there exists a pair of increasing subsequences $I, J$ of length $\ell'$ such that
    \begin{enumerate}
        \item $\rank(G_{I \cap J}) < k$, and
        \item the verification matrix $M_{k,\ell',I,J}$ fails to have full row rank.
    \end{enumerate}
    
    To analyze this failure probability, we introduce three parameters:
	\begin{equation}\label{eq:para_pick}
		\tau=\left\lfloor \frac{\eps n}{8}\right\rfloor, \quad \eps_0=\frac{\eps}{2},\quad r=\left\lfloor \frac{\eps^2 n}{8}\right\rfloor.
	\end{equation}
    We will use the following parameter inequality below:
	\begin{equation}\label{eq:para_condition}
		(1-\eps_0)(\ell'-\tau) \geq h+\frac{r+1}{\eps_0}.
	\end{equation}
	\begin{lemma}\label{lm:parameter_condition}
		Assume that $\ell'\leq n$. For all sufficiently large $n$, the parameters in \eqref{eq:para_pick}, together with $h=2k-1$, satisfy \eqref{eq:para_condition}.
	\end{lemma}
	\begin{proof}
		Since $\ell'=h+\lfloor\eps n\rfloor$ and $\tau=\left\lfloor\eps n/8\right\rfloor$, and since $\ell'\leq n$, we have $\ell'-\tau\leq n$ and $r\leq \eps^2n/8$. Therefore,
		\begin{align*}
			(1-\eps_0)(\ell'-\tau)-h-\frac{r+1}{\eps_0}
			&=\lfloor\eps n\rfloor-\left\lfloor\frac{\eps n}{8}\right\rfloor-\frac{\eps(\ell'-\tau)}{2}-\frac{2(r+1)}{\eps}\\
			&\geq \frac{\eps n}{8}-1-\frac{2}{\eps}.
		\end{align*}
		The last expression is positive for all sufficiently large $n$, which proves \eqref{eq:para_condition}.
	\end{proof}
    
    The following lemma estimates the probability that verification matrix $M_{k,\ell,I,J}$ fails to have full row rank; its proof is deferred to \cref{subsection4.2}.
	
	\begin{lemma}\label{lm:prob_algorithm}
		Let $X\subseteq\overline{\F}_q^k$ be an irreducible MDS variety contained in an affine hyperplane, and let $S\subseteq X(\F_q)$ satisfy $|S|>n$. Fix $0<\eps<1$, set $h=2k-1$ and $\ell'=h+\lfloor\eps n\rfloor$, and define $\tau$, $\eps_0$, and $r$ as in \eqref{eq:para_pick}. Assume that $\ell'\leq n$ and $n$ is sufficiently large. Let $\ell$ satisfy $\ell'-\tau\leq\ell\leq\ell'$, and let $I,J\in[n]^\ell$ be fixed increasing subsequences with $|I\cap J|<k$. If the columns of $G$ are sampled uniformly without replacement from $S$, then the probability that $M_{k,\ell,I,J}[G]$ fails to have full row rank $h$ is at most
		\begin{equation*}
			2^{2k+r-1}\left(\frac{(1/\eps_0+1)P_q(X,2)}{|S|-n}\right)^r,
		\end{equation*}
		where $P_q(X,d)$ denotes the maximum number of $\F_q$-rational points on $X$ that can lie in the zero set of a polynomial with coefficients in $\overline{\F}_q$, of degree at most $d$, and generically non-vanishing on $X$.
	\end{lemma}

    For $k\leq t\leq n$, let $p_t(S)$ denote the probability that a
	$k\times t$ matrix whose columns are sampled uniformly without
	replacement from $S$ fails to have full row rank $k$.

    \begin{theorem}\label{thm:main_prob}
		Let $X$, $S$, and the parameters $\eps$, $\ell'$, $\tau$, $\eps_0$, and $r$ be as in \cref{lm:prob_algorithm}.
		Let $G$ be a $k\times n$ matrix with columns sampled uniformly and without replacement from $S$. Then the code generated by $G$ is $\eps$-relaxed HS-optimal with probability at least
		\begin{equation*}
			1-2^{2n+2k+r} \left(\frac{(1/\eps_0+1)P_q(X,2)}{|S|-n}\right)^r-2^{2n}p_{k+\tau}(S)
		\end{equation*}
	\end{theorem}
	\begin{proof}
		Let $\mathcal{E}$ denote the failure event that the code is not $\eps$-relaxed HS-optimal. For any fixed pair of increasing subsequences $(I, J) \in [n]^{\ell'} \times [n]^{\ell'}$, let $\mathcal{A}_{I,J}$ be the event that $\text{rank}(G_{I \cap J}) < k$, and let $\mathcal{B}_{I,J}$ be the event that the verification matrix $M_{k,\ell',I,J}[G]$ fails to have full row rank $2k-1$.
        If $\rank(G)<k$, then both $\mathcal{A}_{I,J}$ and $\mathcal{B}_{I,J}$ occur for every pair $(I,J)$, so this case is already included in the union of failure events below and requires no separate treatment.
		
		By \cref{lm:condition}, the failure event $\mathcal{E}$ occurs only if both $\mathcal{A}_{I,J}$ and $\mathcal{B}_{I,J}$ occur for at least one pair $(I,J)$. Therefore,
		$$\mathcal{E}\subseteq\bigcup_{(I,J)}(\mathcal{A}_{I,J}\cap\mathcal{B}_{I,J}).$$
		Throughout the proof, $I$ and $J$ range over increasing subsequences in $[n]^{\ell'}$. Partition their pairs into
		\begin{align*}
			\mathcal{P}_1&=\{(I,J):|I\cap J|<k\},\\
			\mathcal{P}_2&=\{(I,J):k\leq|I\cap J|<k+\tau\},\\
			\mathcal{P}_3&=\{(I,J):|I\cap J|\geq k+\tau\}.
		\end{align*}
		Applying the union bound gives
		\begin{align*}
			\Pr(\mathcal{E})\leq {}&\sum_{(I,J)\in\mathcal{P}_1}
			\Pr(\mathcal{A}_{I,J}\cap\mathcal{B}_{I,J})\\
			&+\sum_{(I,J)\in\mathcal{P}_2}
			\Pr(\mathcal{A}_{I,J}\cap\mathcal{B}_{I,J})\\
			&+\sum_{(I,J)\in\mathcal{P}_3}
			\Pr(\mathcal{A}_{I,J}\cap\mathcal{B}_{I,J}).
		\end{align*}
		
		For any pair $(I, J) \in \mathcal{P}_1$, \cref{lm:prob_algorithm} directly bounds $\Pr(\mathcal{B}_{I,J})$, yielding:
		\begin{align*}
			\Pr(\mathcal{A}_{I,J}\cap\mathcal{B}_{I,J})
			&\leq\Pr(\mathcal{B}_{I,J})\\
			&\leq2^{2k+r-1}
			\left(\frac{(1/\eps_0+1)P_q(X,2)}{|S|-n}\right)^r.
		\end{align*}
		
		For any pair $(I,J)\in\mathcal{P}_2$, set $d=|I\cap J|-k+1$. Since $k\leq |I\cap J|<k+\tau$, we have $1\leq d\leq\tau$. Delete from both subsequences the same $d$ positions on which they agree, and denote the resulting subsequences by $I'$ and $J'$. Their length is $\ell=\ell'-d\geq\ell'-\tau$, and $|I'\cap J'|=k-1$. Let $\mathcal{B}'_{I',J'}$ be the event that $M_{k,\ell,I',J'}[G]$ fails to have full row rank $h$. The shortened matrix is obtained from the original matrix by deleting columns, so $\mathcal{B}_{I,J}\subseteq\mathcal{B}'_{I',J'}$. Therefore, \cref{lm:prob_algorithm} applies to $\mathcal{B}'_{I',J'}$ and gives the same probability bound as in Case 1.
		
		The total number of possible sequence pairs in $\mathcal{P}_1$ is bounded by $\binom{n}{\ell'}^2 \le 2^{2n}$, and the number of pairs in $\mathcal{P}_2$ is similarly bounded by $2^{2n}$. Applying the union bound to $\mathcal{P}_1$ and $\mathcal{P}_2$ separately, the combined failure probability for the first two cases is:
		\begin{align*}
			&\sum_{(I,J)\in\mathcal{P}_1}
			\Pr(\mathcal{A}_{I,J}\cap\mathcal{B}_{I,J})
			+\sum_{(I,J)\in\mathcal{P}_2}
			\Pr(\mathcal{A}_{I,J}\cap\mathcal{B}_{I,J})\\
			&\qquad\leq2^{2n+2k+r}
			\left(\frac{(1/\eps_0+1)P_q(X,2)}{|S|-n}\right)^r.
		\end{align*}
		
		For any pair $(I, J) \in \mathcal{P}_3$, we bound the joint probability by the marginal probability of the intersection rank failing: $\Pr(\mathcal{A}_{I,J} \cap \mathcal{B}_{I,J}) \le \Pr(\mathcal{A}_{I,J})$. Because $|I \cap J| \ge k + \tau$, the event $\mathcal{A}_{I,J}$ implies that any subset of $k+\tau$ columns chosen from this intersection must also fail to achieve full rank $k$. This probability is bounded above by the probability that a random $k\times(k+\tau)$ submatrix sampled from $S$ is singular. Thus:
		$$\Pr(\mathcal{A}_{I,J} \cap \mathcal{B}_{I,J}) \le \Pr(\mathcal{A}_{I,J}) \le p_{k+\tau}(S).$$
		
		Applying the union bound over all at most $2^{2n}$ possible pairs in $\mathcal{P}_3$ yields the final term $2^{2n} p_{k+\tau}(S)$. Adding this to the previous sum bounds the total failure probability $\Pr(\mathcal{E})$. The probability that the code succeeds is $1 - \Pr(\mathcal{E})$, which yields the stated bound and completes the proof.
	\end{proof}

    For curves, Bezout's theorem gives explicit bounds for both $P_q(X,2)$ and the rank failure probability $p_{k+\tau}(S)$, leading to the following specialization of \cref{thm:main_prob}.

    \begin{lemma}[Bezout's theorem]\label{lm:bezout}
		Let $X$ be an algebraic curve of degree $d$. If $f$ is a polynomial of degree $d'$ that is generically non-vanishing on $X$, then $f$ has at most $d\cdot d'$ zeros on $X$.
	\end{lemma}

	\begin{corollary}\label{coro:curve}
		Under the assumptions of \cref{thm:main_prob}, suppose in addition that $X$ is an affine curve of degree $d$. Then the code generated by $G$ is $\eps$-relaxed HS-optimal with probability at least
		$$1-2^{2n+2k+r} \left( \frac{2d(1/\eps_0+1)}{|S|-n} \right)^r-2^{2n}\frac{q^k-1}{q-1}\left( \frac{d}{|S|} \right)^{k+\tau}.$$
		Moreover, if $d<k+\tau$, the probability is at least
		$$1-2^{2n+2k+r} \left( \frac{2d(1/\eps_0+1)}{|S|-n} \right)^r.$$
	\end{corollary}
	\begin{proof}
		First, we bound the geometric parameter $P_q(X,2)$. Since $X$ is an irreducible curve of degree $d$, by \cref{lm:bezout}, a polynomial of degree $2$ that does not vanish identically on $X$ has at most $2d$ zeros on $X$. Hence,
		$$P_q(X,2)\leq 2d.$$
		
		Second, we bound the rank failure probability $p_{k+\tau}(S)$. If a $k\times t$ matrix fails to have full row rank, its columns lie in a hyperplane $H\subseteq\F_q^k$ through the origin. Since the MDS curve $X$ is not contained in such a hyperplane, \cref{lm:bezout} gives $|S\cap H|\leq d$. For a fixed $H$, expose the sampled points one at a time. If the first $i$ points all lie in $H$, then the conditional probability that the next point also lies in $H$ is at most
		$$\frac{|S\cap H|-i}{|S|-i}\leq\frac{|S\cap H|}{|S|}\leq\frac{d}{|S|}.$$
		Multiplying these conditional probabilities, the probability that all $t$ sampled points lie in $H$ is at most $(d/|S|)^t$.
		
		There are $(q^k-1)/(q-1)$ hyperplanes through the origin in $\F_q^k$. Taking $t=k+\tau$ and applying the union bound gives
		$$p_{k+\tau}(S)\leq \frac{q^k-1}{q-1}\left(\frac{d}{|S|}\right)^{k+\tau}.$$
		Combining these estimates yields the stated probability bound.
		
		If $d<k+\tau$, no hyperplane can contain $k+\tau$ distinct points. Thus, $p_{k+\tau}(S)=0$ trivially.
	\end{proof}

	\subsection{Chain Decomposition and Rank Certification}\label{subsection4.2}
	
	In this subsection, we prove \cref{lm:prob_algorithm} using the rank-certification framework of Con, Guo, Li, and Zhang \cite{con2024random}. For fixed increasing subsequences $I$ and $J$, the framework decomposes them into small, mutually disjoint chains that cover most of their positions and describe the local relations between the two subsequences. The chains are divided into a primary queue and a reserve bank. The algorithm examines the chains in order and tests whether each one certifies generic non-vanishing of the relevant determinant after the previously exposed variables have been specialized. If a test fails, the current chain is discarded and replaced by a fresh chain from the reserve bank. Consequently, if the rank certification ultimately fails, the algorithm must record a prescribed number of failed chain tests.
	
	The probability estimate follows by exposing the chains one at a time. Because the chains are disjoint, the variables in a fresh chain have not been exposed before and are sampled uniformly from the remaining points of $S$. For a fixed sequence of failed tests, the generic non-vanishing of the determinant and the degree bound in each column allow us to control the conditional probability that each new chain fails in terms of $P_q(X,2)$. Multiplying these conditional bounds gives an upper bound for the probability of the prescribed failure sequence. A union bound over all possible failure sequences and all pairs $(I,J)$ then gives the probability estimate in \cref{lm:prob_algorithm}. Our argument follows the general chain-decomposition and replacement strategy of \cite{con2024random}, adapted to codes obtained from points on an MDS variety. In this setting, the relevant failure probability is controlled through $P_q(X,2)$ by bounding the rational points of $X$ at which a specialization can fail, rather than by counting bad evaluation points in the Reed-Solomon setting.
    
    We now introduce the definitions needed for the chain decomposition.
	
	\begin{definition}[Subsequence Selection]
		For an increasing subsequence $I\in [n]^{\ell}$ and $P=\{i_1,\dots,i_{\ell'}\}\subseteq [\ell]$ with $i_1<\dots<i_{\ell'}$, denote by $I^P$ the increasing subsequence $(I_{i_1},\dots,I_{i_{\ell'}})$.
	\end{definition}
	
	The fundamental unit of our analysis is the chain, which captures the local coupling between two subsequences.
	\begin{definition}[Chain]
		Let $I, J \in [n]^\ell$ be two increasing subsequences. We say the pair $(I, J)$ is a \emph{chain} if one of the following conditions holds:
		\begin{itemize}
			\item Type I: $\ell=1$ and $I_1 =J_1$.
			\item Type II: $\ell\geq2$ and either $I_i=J_{i+1}$ for all $1 \leq i<\ell$, or $I_{i+1} =J_i$ for all $1\leq i<\ell$.
		\end{itemize}
	\end{definition}
	
	For the replacement procedure, we need chains that do not share variables.
	
	\begin{definition}[Indices of variables]
		For increasing subsequences $I,J\in [n]^{\ell}$ and subset $P\subseteq [\ell]$, we define the set of variable indices involved in the block corresponding to $P$ as:
		$$\Var(I,J,P) \triangleq \Set(I^P) \cup \Set(J^P) \subseteq [n].$$ 
	\end{definition}
	
	\begin{definition}[$(I,J)$-disjointness]
		Let $I,J\in[n]^\ell$ be increasing subsequences. Two sets $P,P'\subseteq [\ell]$ are said to be \emph{$(I,J)$-disjoint} if the sets of variables $\Var(I,J,P)$ and $\Var(I,J,P')$ are disjoint.
	\end{definition}
	
	The following result gives disjoint chains of bounded size.
	
	\begin{lemma}[\cite{con2024random}, Lemma 30]\label{lm:split}
		Let $I, J\in [n]^\ell$ be increasing subsequences. Let $\eps_0 \in (0,1)$.
		Then there exist a subset $P\subseteq [\ell]$ of size at least $(1-\eps_0)\ell$ and a partition
		$P=P_1\sqcup P_2\sqcup \dots\sqcup P_s$ into nonempty sets, where $\sqcup$ denotes a disjoint union, such that
		\begin{enumerate}
			\item Each $(I^{P_i},J^{P_i})$ is a chain.
			\item The sets $P_1,\dots,P_s$ are mutually $(I,J)$-disjoint.
			\item $|P_i|\leq 1/\eps_0$ for every $i\in [s]$.
		\end{enumerate}
	\end{lemma}

    Permute $P_1,\dots,P_s$, if necessary, so that $|P_1|\leq\dots\leq|P_s|$ and all Type I chains precede the Type II chains. Thus, if $i<j$ and $(I^{P_j},J^{P_j})$ is Type I, then $(I^{P_i},J^{P_i})$ is also Type I. This ordering ensures that a chain taken from the bank is at least as large as the chain it replaces.

    We now explain the role of \cref{eq:para_condition}. The chains retained by \cref{lm:split} contain at least $(1-\eps_0)\ell$ positions. We reserve the last $m$ chains as a bank, choosing $m$ so that their total size first reaches $r/\eps_0$. The next lemma shows that the bank contains enough chains for $r$ possible replacements, while the remaining chains still contain enough positions to form the required $h$ verification columns.
    
    \begin{lemma}\label{lm:sum-of-Pi}
		For the fixed parameters above and all sufficiently large $n$, there exists an index $m\in [s]$ such that
		$$\sum_{i=s-m+1}^s |P_i|\in [r/\eps_0, (r+1)/\eps_0].$$
		Furthermore, the remaining chains satisfy
		$$\sum_{i=1}^{s-m}|P_i|\geq h$$
		and we have $m\geq r$.
	\end{lemma}
	\begin{proof}
	Since $\ell\geq\ell'-\tau$, the chains cover at least $(1-\eps_0)\ell\geq(1-\eps_0)(\ell'-\tau)$ positions. Thus, \cref{eq:para_condition} ensures that their total size exceeds $r/\eps_0$. Choose the last $m$ chains greedily until their total size first reaches $r/\eps_0$. Since each chain has size at most $1/\eps_0$, the selected total lies in $[r/\eps_0,(r+1)/\eps_0]$. Consequently,
	$$\sum_{i=1}^{s-m}|P_i|\geq (1-\eps_0)(\ell'-\tau)-\frac{r+1}{\eps_0}\geq h.$$
	Moreover, the selected chains have total size at most $m/\eps_0$, so their total size being at least $r/\eps_0$ implies $m\geq r$.
	\end{proof}
	
	To construct our verification framework, we must define the specific submatrices that the algorithm will evaluate.
	\begin{definition}\label{def:block_matrix}
		Let $n,k,\ell,s\in\N^+$, let $I,J\in[n]^\ell$ be increasing subsequences, and let $G$ be a $k\times n$ matrix. Let $P_1,\dots,P_s\subseteq [\ell]$ be pairwise disjoint nonempty sets such that each $(I^{P_i},J^{P_i})$ is a chain. Define the corresponding block matrix by horizontal concatenation:
		$$\widetilde{M}_{k,I,J,(P_i)_{i=1}^s}[G]
		=\begin{pmatrix}
			M_{k,|P_1|,I^{P_1},J^{P_1}}[G] &
			M_{k,|P_2|,I^{P_2},J^{P_2}}[G] &
			\cdots &
			M_{k,|P_s|,I^{P_s},J^{P_s}}[G]
		\end{pmatrix}.$$
	\end{definition}
	
	\begin{lemma}\label{lm:submatrix}
		Under the notation of \cref{def:block_matrix}, if $\widetilde{M}_{k,I,J,(P_i)_{i=1}^s}[G]$ has full row rank, then $M_{k,\ell,I,J}[G]$ also has full row rank.
	\end{lemma}
	\begin{proof}
		The block matrix is obtained from $M_{k,\ell,I,J}[G]$ by selecting and possibly reordering columns. Hence its row rank cannot exceed that of $M_{k,\ell,I,J}[G]$.
	\end{proof}

    We now give the rank-certification procedure. Its pseudocode is presented in Algorithm~\ref{algorithm}. We subsequently explain its initialization, chain verification, and termination.

	\begin{algorithm}[t]
		\caption{Rank Certification Procedure}
		\label{algorithm}
		\DontPrintSemicolon
		\KwIn{$n,k,\ell,r\in\N^+$, $h=2k-1$, an irreducible variety $X$ contained in an affine hyperplane, $(\alpha_1,\dots,\alpha_n)\in X^n$, $\eps_0\in (0,1)$, increasing subsequences $I,J\in[n]^{\ell}$, and ordered sets $P_1,\dots,P_s$ obtained from \cref{lm:split}.}  
		\KwOut{``SUCCESS'' or a sequence $(i_1,\dots,i_r)\in \{0,\dots,h-1\}^r$.}  
		\BlankLine
		$\cA\gets \emptyset$. \tcp*[f]{set of indices of assigned variables} \label{line:1}
		
		$c\gets 0$. \tcp*[f]{number of assigned columns} \label{line:2}
		
		Pick the smallest $m\in [s]$ such that $\sum_{i=s-m+1}^s |P_i|\in [r/\eps_0, (r+1)/\eps_0]$. \tcp*[f]{last $m$ chains form the bank} \label{line:3}
		
		$a\gets 1$. \tcp*[f]{current chain to assign} \label{line:4}
		
		$b\gets s-m+1$. \tcp*[f]{current chain in the bank} \label{line:5}
		
		$j\gets 0$. \tcp*[f]{number of failed attempts} \label{line:6}
		
		\lFor{$i=1$ to $s$}{$Q_i\gets P_i$.} \label{line:7}
		
		\While{$c<h$\label{line:8}} 
		{
			$\Delta\gets$ the determinant of the left $h\times h$ submatrix of $\widetilde{M}_{k,I,J,(Q_i)_{i=1}^{s-m}}[G]$. \label{line:9}
			
			$\cA'\gets \cA\cup\Var(I,J,Q_a)$. \label{line:10}
			
			$\overline{\Delta}\gets \Delta|_{G_i=\alpha_i \text{ for } i\in \cA'}$. \label{line:11}
			
			\tcp*[r]{assign the variables in the $a$-th chain}
			\uIf{$\overline{\Delta}$ is generically non-vanishing in the remaining variables\label{line:12}}
			{
				$\cA\gets \cA'$. \label{line:13}
				
				$c\gets c + |Q_a|$. \label{line:14}
				
				$a\gets a+1$. \label{line:15}
			}
			\Else(\tcp*[f]{replace the $a$-th chain by (part of) the $b$-th chain in the bank})
			{
				
				$Q_a\gets$ set of the smallest $|Q_a|$ elements in $Q_{b}$.
				\label{line:17} 
				
				$Q_b\gets \emptyset$. \label{line:18}
				
				$b\gets b+1$. \label{line:19}

				$j\gets j+1$. \label{line:20}
				
				$i_j\gets c$. \label{line:21}
				
				\lIf{$j=r$}{output $(i_1,\dots,i_r)$ and halt. \label{line:22}}
			}
		}\label{line:24}
		Output ``SUCCESS''. \label{line:25}
	\end{algorithm}
    \medskip
	
	The algorithm has three steps:
	\begin{enumerate}
		\item \textbf{Initialization:} Choose $m$ as in Line~\ref{line:3}. The first $s-m$ chains form the primary queue, and the last $m$ chains form the bank used for replacements.
		
		\item \textbf{Chain verification:} For the current chain $Q_a$, specialize $G_i=\alpha_i$ for every $i\in\Var(I,J,Q_a)$ and test whether the chosen determinant is generically non-vanishing in the remaining variables (Lines~\ref{line:9}-\ref{line:12}). If the test succeeds, the algorithm accepts this chain, increases $c$ by the number $|Q_a|$ of verification columns in its block, and advances to the next chain. If the test fails, the specialized determinant is identically zero in the remaining variables; the algorithm records $i_j=c$ and replaces $Q_a$ by a chain of the same size taken from the bank (Lines~\ref{line:17}-\ref{line:22}).
		
		\item \textbf{Termination:} The algorithm outputs ``SUCCESS'' once $c\geq h$. If the number of failed tests reaches $r$ first, it outputs $(i_1,\dots,i_r)$.
	\end{enumerate}

    The algorithm above is the transpose analogue of the rank-certification algorithm in \cite{con2024random}. There, the relevant matrix has one row for each position $t\in[\ell]$, whose entries are the powers of the scalar variables $X_{I_t}$ and $X_{J_t}$, and the algorithm assigns these scalar variables chain by chain. Here, $M_{k,\ell,I,J}[G]$ has one column for each position $t\in[\ell]$, formed from the vector variables $G_{I_t}$ and $G_{J_t}$; these generator-matrix columns are points of $X$. Thus, our procedure successively assigns columns $G_i=\alpha_i$ and considers the left $h\times h$ submatrix, whereas \cite{con2024random} assigns the $X_i$ and considers the top $h\times h$ submatrix. Apart from this transposition and the different variables being assigned, the chain decomposition, determinant checks, and replacement steps are the same.

	The correctness of \cref{algorithm} follows from the remaining properties below. They depend only on the sequence indices, so their proofs use the same combinatorial arguments as in \cite{con2024random}.
	
	The following lemma verifies the validity of the steps of \cref{algorithm} after its initialization. In particular, \cref{item:b} ensures that Line~\ref{line:9} is valid, while \cref{item:c,item:d} ensures that Lines~\ref{line:17}-\ref{line:18} are valid.
	\begin{lemma}\label{lm:basic-claims}
		During \cref{algorithm}, after Line~\ref{line:7}:
		\begin{enumerate}[(a)]
			\item\label{item:a}
			$|Q_i|=|P_i|$ for $i\in [s-m]$. In particular, $|Q_i|\leq 1/\eps_0$ for $i\in [s-m]$.
			\item\label{item:b}
			$\sum_{i=1}^{s-m}|Q_i|\geq h$.
			\item\label{item:c0}
			$1\leq a\leq s-m$ always holds at Line~\ref{line:9}.
			\item\label{item:c}
			$s-m+1\leq b\leq s$ always holds at Line~\ref{line:9}.
			\item\label{item:d}
			$|Q_{b}|\geq |Q_a|$ always holds at Line~\ref{line:9}, and consequently, also at Line~\ref{line:17}.
		\end{enumerate}
	\end{lemma}
	\noindent \textit{Proof Sketch.} These invariants hold dynamically because the algorithm only ever shrinks or discards chains during replacements. In particular, the inequality $|Q_b|\ge |Q_a|$ follows from the initial ordering of the chains by size: chains taken from the bank, indexed by $b$, are always at least as large as the corresponding chains in the primary queue, indexed by $a$.
    \medskip
	
	\begin{lemma}\label{lm:terminate} 
		\cref{algorithm} terminates, outputting either ``SUCCESS'' or a sequence  $(i_1,\dots,i_r)\in\{0,\dots,h-1\}^r$. In the former case, the specialized matrix $M_{k,\ell,I,J}|_{G=\valpha}$ has full row rank. 
	\end{lemma}
	\noindent \textit{Proof Sketch.} The variable $c$ records the number of assigned columns and strictly increases after each successful assignment, namely, when $\overline{\Delta}$ is generically non-vanishing in the remaining variables. If $c$ reaches $h$, then all variables involved in the chosen $h\times h$ minor have been assigned and the resulting value of $\overline{\Delta}$ is nonzero. Hence, $M_{k,\ell,I,J}|_{G=\valpha}$ contains a nonsingular $h\times h$ submatrix and therefore has full row rank. On the other hand, whenever the rank-certification step fails to progress, the failure counter $j$ strictly increases; thus the algorithm must terminate once $j$ reaches the prescribed budget $r$.
    \medskip

	As the variable $c$ in \cref{algorithm} never decreases, we have the following lemma.
	\begin{lemma}\label{lm:output}
		Any sequence $(i_1,\dots,i_r)$ output by \cref{algorithm} satisfies $i_1\leq \dots\leq i_r$.
	\end{lemma}
	
	Recall that $\Var(I,J,Q_i)$ denotes the set of variable indices involved in $M_{k,|Q_i|,I^{Q_i},J^{Q_i}}$.
    
	\begin{lemma}\label{lm:disjoint}
		The sets $Q_1,\dots,Q_s$ are mutually $(I,J)$-disjoint before and after the replacement operations in Lines~\ref{line:17}-\ref{line:18}.
	\end{lemma}
	\noindent \textit{Proof Sketch.} They are initially disjoint by the chain decomposition theorem. Each replacement uses a subset of a fresh bank chain, which is disjoint from all chains currently in use, so disjointness is preserved.
    \medskip

	\begin{lemma}\label{lm:new-elements}
		At Line~\ref{line:9} of \cref{algorithm}, none of the points indexed by $\Var(I,J,Q_a)$ has been exposed in a previous iteration of the while loop. In particular, none of these indices has previously been added to $\cA'$.
	\end{lemma}
	\noindent \textit{Proof Sketch.} This follows directly from Lemma~\ref{lm:disjoint}: once a chain is evaluated, it is either accepted or discarded, and the next chain considered is disjoint from all previously used chains. Hence all points indexed by $\Var(I,J,Q_a)$ are unexposed.
    \medskip
    
	The following lemma is a key ingredient in the analysis. It shows that, once $j$, the failure sequence $(i_1,\dots,i_j)$, and the current value of $c$ are fixed, the combinatorial state of the algorithm is determined independently of the actual values of $\alpha_1,\dots,\alpha_n$.
	
	\begin{lemma}\label{lm:determined}
		At Line~\ref{line:12} of \cref{algorithm}, the values of $Q_1,\dots,Q_s$ and $a$ are determined by the values of $j$, $i_1,\dots,i_j$, and $c$, together with the input excluding $\alpha_1,\dots,\alpha_n$. 
	\end{lemma}
	\noindent \textit{Proof Sketch.} The pointers $a$ and $b$, as well as the queues $Q_1,\dots,Q_s$, are updated solely according to the outcomes of the determinant checks. More precisely, a successful check increases $c$ and advances $a$, while a failed check records the current value of $c$, replaces the current queue by a prescribed subset of the next bank queue, and advances $b$. Therefore, once the previous failure sequence $(i_1,\dots,i_j)$ and the current value of $c$ are fixed, the entire combinatorial state of the algorithm is fixed as well. This state depends only on the input combinatorial data and not on the actual values of $\alpha_1,\dots,\alpha_n$. For a detailed proof, we refer the reader to \cite[Proof of Lemma~39]{con2024random}.
    \medskip

	We now justify the algebraic verification step. We first give a simple degree bound for the determinants of the relevant minors.
	
	\begin{lemma}\label{lm:degree}
		Let $X\subseteq\overline{\F}_q^k$ be an irreducible MDS variety contained in an affine hyperplane, and set $h=2k-1$. Let $I,J\in[n]^\ell$ be increasing subsequences that agree on at most $k-1$ coordinates. For every subset $P\subseteq[\ell]$ of size $h$, the determinant
		$$\det(M_{k,h,I^P,J^P}[G])$$
		is generically non-vanishing when its column variables are chosen from $X$. Moreover, its degree in each column variable $G_i$ is at most $2$.
	\end{lemma}
	\begin{proof}
		By \cref{thm:Variety}, the MDS variety $X$ is HS-optimal. Since $X$ is contained in an affine hyperplane and $I^P,J^P$ agree on at most $k-1$ coordinates, HS-optimality gives the generic non-vanishing of $\det(M_{k,h,I^P,J^P}[G])$.

		A fixed generator column $G_i$ occurs in the verification matrix at most twice: once through $I^P$ and once through $J^P$. The determinant is multilinear in the columns of the verification matrix, so its degree in the entries of $G_i$ is at most $2$.
	\end{proof}
	
	We next show that the determinant remains generically non-vanishing after a chain is replaced.
	
	\begin{lemma}\label{lm:nonsingular}
		At Line~\ref{line:12} of \cref{algorithm}, the determinant $\Delta|_{G_i=\alpha_i\text{ for }i\in\cA}$,	viewed as a polynomial in the unassigned column variables $G_i$, $i\notin\cA$, is generically non-vanishing when these variables range over $X$.
	\end{lemma}

\begin{proof}
		The proof follows the re-indetermination argument in \cite[Proof of Lemma~40]{con2024random}. We argue by induction on the iterations of the while loop. In the first iteration, $\cA=\emptyset$. Up to a permutation of columns, $\Delta$ is the determinant of a matrix of the form $M_{k,h,I^P,J^P}[G]$ for some $P\subseteq[\ell]$ of size $h$. By \cref{lm:degree}, this determinant is generically non-vanishing.

		Suppose that, at Line~\ref{line:12} in a given iteration, the partially specialized determinant $\Delta|_{G_i=\alpha_i,\ i\in\cA}$ is generically non-vanishing in the remaining variables. If the test at Line~\ref{line:12} succeeds, then the algorithm replaces $\cA$ by
		$$\cA'=\cA\cup\Var(I,J,Q_a),$$
		and the test itself asserts that the resulting determinant is generically non-vanishing. Thus, the same property holds in the next iteration.

		It remains to consider the case in which the test fails. Then $\cA$ is unchanged, but the current chain $Q_a$ is replaced by a subset $Q_b'$ of a fresh bank chain. By \cref{lm:new-elements}, none of the variables indexed by $\Var(I,J,Q_a)$ or $\Var(I,J,Q_b')$ has been assigned before. Let $\Delta_{\mathrm{old}}$ and $\Delta_{\mathrm{new}}$ denote the partially specialized determinants before and after this replacement, respectively.

		The two determinants differ only in the block corresponding to the replaced chain. As in \cite[Proof of Lemma~40]{con2024random}, there are three cases: the two chains are Type II with the same orientation, the two chains are Type II with opposite orientations, or the old chain is Type I. In the first case, the old block is obtained from the new block by relabeling the fresh column variables. In the second case, the same holds after reversing the order of the new chain, which only permutes columns and therefore changes the determinant by a sign. In the Type I case, the same conclusion follows by identifying the two variables in the new block with the single variable in the old Type I block. Finally, the ordering of the chains implies that, if the new bank chain is Type I, then the old chain is also Type I. Hence these cases cover all possibilities.

		Consequently, there exists a relabeling homomorphism $\sigma$ of the remaining column variables, possibly together with a column permutation, such that $\Delta_{\mathrm{old}}=\pm\sigma(\Delta_{\mathrm{new}})$.
		If $\Delta_{\mathrm{new}}$ were identically zero on the product of the remaining copies of $X$, then so would be $\Delta_{\mathrm{old}}$, contradicting the induction hypothesis. Thus, $\Delta_{\mathrm{new}}$ is generically non-vanishing. The invariant is therefore preserved at every iteration, which proves the lemma.
	\end{proof}

	We now combine these algebraic guarantees with the geometric density parameter $P_q(X, d)$ to bound the probability of a specific failure sequence. This allows us to formally prove the fundamental bound introduced in \cref{subsection:probabilistic_analysis}.
	
	\begin{proof}[Proof of Lemma~\ref{lm:prob_algorithm}]
		By \cref{lm:terminate}, if the matrix fails to achieve full row rank, \cref{algorithm} must output a sequence $(i_1, \dots, i_r) \in \{0, \dots, h-1\}^r$. Let
		$$p=\frac{(1/\eps_0+1)P_q(X,2)}{|S|-n}.$$
		For $0 \le t \le r$, let $\mathcal{F}_t$ be the event that the algorithm matches a fixed target sequence up to the $t$-th index. We prove by induction that $\Pr[\mathcal{F}_t] \le p^t$.
		
		For $t=0$, the claim holds trivially. Assume the bound holds for some $t<r$. We now prove that it holds for $t+1$ as well. Denote the fixed target sequence by $(i_1^*,\dots,i_r^*)$. Let $\mathcal{F}_t'$ be the sub-event of $\mathcal{F}_t$ that the algorithm reaches Line~\ref{line:12} with $j=t$, $(i_1,\dots,i_t)=(i^*_1,\dots,i^*_t)$, and $c=i^*_{t+1}$.
		Note that $\mathcal{F}_{t+1}$ is a sub-event of $\mathcal{F}_t'$. If $\Pr[\mathcal{F}_t'] = 0$, then $\Pr[\mathcal{F}_{t+1}] = 0$ and the step is proven. Otherwise, by the inductive hypothesis, $\Pr[\mathcal{F}_t'] \le \Pr[\mathcal{F}_t] \le p^t$. It therefore suffices to bound the conditional probability $\Pr[\mathcal{F}_{t+1} \mid \mathcal{F}_t'] \le p$. 
		
		Condition on the event $\mathcal{F}_t'$. Consider the moment when the algorithm reaches Line~\ref{line:12} with $j=t$, $(i_1,\dots,i_t)=(i^*_1,\dots,i^*_t)$, and $c=i^*_{t+1}$.
		By \cref{lm:determined}, $Q_1,\dots,Q_s$ and the index $a$ are entirely determined. 
		The set of assigned variables $\mathcal{A} = \bigcup_{i=1}^{a-1} \Var(I,J,Q_i)$ is also determined. 
		
		Let $E$ be the set of all point indices exposed before the current test, including those in discarded chains. Here, the exposed points include those in chains accepted in previous iterations and those in chains discarded after failed tests. Further condition on arbitrary values of $\alpha_i$, $i\in E$, that are consistent with $\mathcal F_t'$. By \cref{lm:new-elements}, $E$ is disjoint from $\Var(I,J,Q_a)$. Therefore, conditioned on these previously exposed values, the points indexed by $\Var(I,J,Q_a)$ are sampled uniformly without replacement from the remaining points of $S$. The following bound is uniform over all choices of the conditioned values, so it also holds after averaging over them.
		
		Let $\Delta^*:=\Delta|_{G_i=\alpha_i \text{ for } i\in \cA}$. After the current chain is exposed, the determinant tested at Line~\ref{line:12} is the polynomial $\overline{\Delta}:=\Delta^*|_{G_i=\alpha_i \text{ for } i\in \Var(I,J,Q_a)}$ in the other unassigned column variables. By \cref{lm:nonsingular}, $\Delta^*$ is generically non-vanishing. List the indices in $\Var(I,J,Q_a)$ as $u_1,\ldots,u_d$. Since $(I^{Q_a},J^{Q_a})$ is a chain, $d\leq |Q_a|+1\leq1/\eps_0+1$. Moreover, by \cref{lm:degree}, $\Delta^*$ has degree at most $2$ in each column variable $G_{u_j}$.

		Expose $\alpha_{u_1},\ldots,\alpha_{u_d}$ one at a time. Before exposing $\alpha_{u_j}$, apply \cref{lm:poly_degenerate} with $G_{u_j}$ as the first group of variables and all remaining unassigned columns as the second group. As long as the previous specializations preserve generic non-vanishing, this gives a polynomial $g_j$ on $X$ that is generically non-vanishing and has degree at most $2$. If $g_j(\alpha_{u_j})\neq0$, specializing $G_{u_j}=\alpha_{u_j}$ again preserves generic non-vanishing. Conditioned on the previously exposed points, $\alpha_{u_j}$ is uniform over at least $|S|-n$ available points. Hence the probability that this specialization fails is at most
		$$\frac{P_q(X,2)}{|S|-n}.$$
		If every specialization preserves generic non-vanishing, then $\overline{\Delta}$ remains generically non-vanishing in the other unassigned column variables. Therefore, for $\overline{\Delta}$ to be identically zero, at least one of these $d$ specializations must fail. A union bound over the $d\leq1/\eps_0+1$ exposed points therefore gives
		$$\Pr[\overline{\Delta}\text{ is identically zero}\mid\mathcal{F}_t']\leq
		\frac{(1/\eps_0+1)P_q(X,2)}{|S|-n}=p.$$
		The event $\mathcal F_{t+1}$ requires this test to fail, so $\Pr[\mathcal F_{t+1}\mid\mathcal F_t']\leq p$, completing the induction.
		
		By \cref{lm:terminate}, if the matrix $M_{k,\ell,I,J}|_{G_1=\alpha_1, \dots, G_n=\alpha_n}$ does not have full row rank, then \cref{algorithm} outputs a sequence $(i_1,\dots,i_r)\in\{0,\dots,h-1\}^r$. By \cref{lm:output}, any output sequence satisfies $i_1\leq \dots\leq i_r$. The number of such sequences is $\binom{h+r-1}{r}\leq 2^{h+r-1}\leq 2^{2k+r-1}$.
		Using the union bound on all such sequences, the failure probability is at most 
		$$2^{2k+r-1}\left(\frac{(1/\eps_0+1)P_q(X,2)}{|S|-n}\right)^r.$$ 
	\end{proof}

	\section{Application to AG Codes}\label{section5}
	
	We now apply our probabilistic framework to algebraic geometry codes constructed from curves over finite fields. After recalling the standard construction of AG codes, we show that their generator matrices can be viewed as matrices whose columns are sampled from an affine MDS curve. This allows us to apply the probabilistic bounds from \cref{thm:main_prob,coro:curve} and derive explicit results for several families of curves.
	
	To define algebraic geometry codes, we use projective varieties. Let $\bbP^n$ denote the $n$-dimensional projective space over $\F$, whose points are written in homogeneous coordinates $[x_0:x_1:\cdots:x_n]$, where two nonzero vectors define the same point if they differ by a nonzero scalar multiple. A \emph{projective variety} in $\bbP^n$ is the common zero set of homogeneous polynomials in $\F[x_0,\dots,x_n]$, or equivalently, the zero set of a homogeneous ideal.

	Every affine variety $X\subseteq\mathbb{A}^n$ can be embedded into projective space by the standard map
	$$(a_1,\dots,a_n)\mapsto[1:a_1:\cdots:a_n].$$
	The \emph{projective closure} of $X$, denoted by $\overline X$, is the Zariski closure of this image in $\bbP^n$. Concretely, if $X$ is defined by polynomials in $\F[x_1,\dots,x_n]$, then $\overline X$ is obtained by homogenizing these defining equations with respect to $x_0$. The points in $\overline X\setminus X$ are called points at infinity.

	We briefly recall the standard setup for AG codes. An algebraic curve is an irreducible variety of dimension $1$. It is \emph{smooth} if all its points are non-singular, and \emph{complete} if it is proper (in the case of curves, this is equivalent to being a projective variety). In this section, let $X$ be a smooth complete geometrically irreducible curve of genus $g$ defined over $\F_q$, and let $\F_q(X)$ denote its function field.

	Divisors and Riemann-Roch spaces can be defined over arbitrary fields, and their geometric properties are often studied after extending the base field to an algebraic closure. Here, however, our goal is to construct AG codes over $\F_q$. We therefore work with divisors defined over $\F_q$ and define their Riemann-Roch spaces using rational functions in $\F_q(X)$. We pass to $\overline{\F}_q$ only when viewing the generator-matrix columns as points on an algebraic variety.
	
	For a rational function $f\in \F_q(X)$ and a point $P\in X(\overline{\F}_q)$, the \emph{discrete valuation} of $f$ at $P$, denoted by $v_{P}(f)$, is defined as the order of the zero (if positive) or pole (if negative) of $f$ at $P$.
	
	A \emph{divisor} $D$ on $X$ is a formal finite sum of the form $D=\sum a_P P$, where $P\in X(\overline{\F}_q)$ and $a_P\in\mathbb{Z}$. A divisor is said to be \emph{defined over $\F_q$} if it is invariant under the Frobenius automorphism. Throughout this section, all divisors are assumed to be defined over $\F_q$. The set of divisors on $X$ is denoted by $\Div(X)$. The \emph{support} of $D$, denoted by $\supp(D)$, is the set $\{P\in X(\overline{\F}_q):a_P\neq0\}$. Its \emph{degree} is $\deg(D)=\sum a_P$. If every coefficient $a_P$ is non-negative, then $D$ is called an \emph{effective divisor}, denoted by $D\geq 0$.

	For a nonzero rational function $f\in\F_q(X)$, its \emph{principal divisor} is $(f)=\sum_Pv_P(f)P$. The \emph{Riemann-Roch space} associated with a divisor $D$ is
	$$L(D)=\{0\}\cup\{f\in\F_q(X):(f)+D\geq0\}.$$
	It is a finite-dimensional vector space over $\F_q$, whose dimension is denoted by $\ell(D)$. The Riemann-Roch theorem also applies to divisors defined over $\F_q$ and to the $\F_q$-vector spaces $L(D)$ considered here.
	\begin{lemma}[Proposition 1.3.4 in Chapter IV, \cite{hartshorne2013algebraic}]
		Let $X$ be a smooth complete curve of genus $g$, and let $D$ be a divisor on $X$. If $\deg(D)>2g-2$, then $\ell(D)=\deg(D)-g+1$.
	\end{lemma}
	
	Let $\cP=(P_1,\dots,P_n)$ be a tuple of distinct $\F_q$-rational points on $X$ such that $\supp(D)\cap \cP=\emptyset$ and $n>\deg(D)$. The evaluation map
	$$\Ev_{\cP}:L(D)\to \F_q^n, \quad f \mapsto (f(P_1),\dots,f(P_n)).$$
	defines the AG code $\mathcal{C}_L(X, \mathcal{P}, D)$ as its image. Its kernel is $L(D-P_1-\cdots-P_n)$. This space is zero because $\deg(D-P_1-\cdots-P_n)<0$: a nonzero function $f$ in this space would make $(f)+D-P_1-\cdots-P_n$ an effective divisor. Since every principal divisor has degree zero, this effective divisor would have negative degree, which is impossible. Hence the evaluation map is injective, and $\mathcal{C}_L(X, \mathcal{P}, D)$ is an $[n,k]_q$ linear code with $k=\ell(D)\leq n$.

	\begin{lemma}\label{lm:AG_all_one}
		Write $D=\sum_P a_PP$ and let $D_+=\sum_P\max\{a_P,0\}P$. If $n>\deg(D_+)$, then $\mathbf{1}\in\mathcal{C}_L(X,\mathcal{P},D)$ if and only if $D$ is effective. 
	\end{lemma}
	\begin{proof}
		If $D$ is effective, then $1\in L(D)$, so its evaluation is $\mathbf{1}$. Conversely, suppose that $f\in L(D)$ evaluates to $\mathbf{1}$, and set $h=f-1$. The poles of $h$ are bounded by $D_+$, and $h$ vanishes at $P_1,\dots,P_n$. Hence
		$$h\in L(D_+-P_1-\cdots-P_n)=\{0\},$$
		where the equality follows from $n>\deg(D_+)$. Thus $f=1$. Since $1\in L(D)$, the divisor $D$ is effective.
	\end{proof}

    For simplicity, we restrict to one-point divisors of the form $D=dP$, where $P\in X(\F_q)$ and $d$ is a non-negative integer. This restriction is sufficient for our purposes, since one-point divisors already provide the families of AG codes needed for the results below; no optimization over the choice of divisor is required.
    
	To connect this standard construction to the analysis in \cref{Section4}, we show that $\mathcal{C}_L(X, \mathcal{P}, D)$ can be realized directly as a code generated by sampling points from an affine curve $Y$ defined over $\F_q$.

	\begin{lemma}\label{lm:code_equiv}
		Let $X$ be a smooth complete curve of genus $g$ defined over $\F_q$, let $P\in X(\F_q)$, and let $D=dP$ with $d\geq 2g+1$. Set $k=\ell(D)$ and choose a basis $f_0=1,f_1,\dots,f_{k-1}$ of $L(D)$. Then the map
		$$\psi:X\setminus\{P\}\longrightarrow\overline{\F}_q^k,
		\qquad Q\longmapsto(1,f_1(Q),\dots,f_{k-1}(Q))$$
		has image an irreducible MDS affine curve $Y$ of degree $d$. For every tuple $\cP=(P_1,\dots,P_n)$ of distinct points in $X(\F_q)\setminus\{P\}$, the vectors $\psi(P_1)^T,\dots,\psi(P_n)^T$ are the columns of a generator matrix of $\cC_L(X,\cP,D)$. Moreover, for every $S\subseteq X(\F_q)\setminus\{P\}$, the map $\psi$ gives a bijection from $S$ to $S'=\psi(S)\subseteq Y(\F_q)$. Thus, sampling points uniformly without replacement from $S$ is equivalent to sampling columns uniformly without replacement from $S'$.
	\end{lemma}
	
	The proof uses standard results about embeddings defined by linear systems; for completeness, we give the details in \cref{appendix:code_equiv}.
	
	Because \cref{lm:code_equiv} realizes these AG codes as codes generated by points on an irreducible MDS affine curve, \cref{thm:Variety} gives the following consequence.
	\begin{corollary}
		Under the assumptions of \cref{lm:code_equiv}, let $n\geq2k$. For a generic tuple of evaluation points $\cP=(P_1,\dots,P_n)$, the code $\cC_L(X,\cP,D)$ is HS-optimal.
	\end{corollary}
	We now apply the probabilistic bound from \cref{coro:curve} to AG codes.
	
	\begin{theorem}\label{thm:AG_code}
		Let $X$ be a smooth complete curve of genus $g$ defined over $\F_q$, let $P\in X(\F_q)$, and let $D=dP$ with $d\geq2g+1$. Set $k=\ell(D)=d-g+1$, and let $S\subseteq X(\F_q)\setminus\{P\}$ satisfy $|S|>n>d$. Fix $0<\eps<1$, set $\ell'=2k-1+\lfloor\eps n\rfloor$, and define $\tau,r,\eps_0$ as in \eqref{eq:para_pick}. Suppose that $g\leq\tau$, $\ell'\leq n$, $r\geq1$, and $n$ is sufficiently large.
		Let $\cP=(P_1,\dots,P_n)$ be sampled uniformly without replacement from $S$. Then $\cC_L(X,\cP,D)$ is $\eps$-relaxed HS-optimal with probability at least
		$$1-2^{2n+2k+r}\left(\frac{2d(1/\eps_0+1)}{|S|-n}\right)^r.$$
	\end{theorem}
	\begin{proof}
		By \cref{lm:code_equiv}, the chosen basis of $L(D)$ gives an irreducible MDS affine curve $Y\subseteq\overline{\F}_q^k$ of degree $d$. It also gives a bijection between $S$ and $S'=\psi(S)\subseteq Y(\F_q)$ under which the sampled points become the columns of a generator matrix of $\cC_L(X,\cP,D)$. In particular, sampling from $S$ is equivalent to sampling columns uniformly without replacement from $S'$.
		
		Since $Y\subseteq\{x_1=1\}$, it is contained in an affine hyperplane. All the assumptions of \cref{thm:main_prob}, and hence of \cref{coro:curve}, are therefore satisfied for $Y$ and $S'$. Finally,
		$$\deg(Y)=d=k+g-1<k+\tau,$$
		where the strict inequality follows from $g\leq\tau$. The second bound in \cref{coro:curve} now gives the stated probability. Notice that the argument applies the corollary directly to $Y$; no comparison between the parameters $P_q(X,2)$ and $P_q(Y,2)$ is needed.
	\end{proof}

	The probability bound in \cref{thm:AG_code} depends on the ratio $\deg(D)/(|S|-n)$, while its degree condition requires $g\leq\tau=\Theta(n)$. Since $\deg(D)=k+g-1$, we need curves whose genus is not too large relative to the block length and that still provide many more than $n$ rational points for sampling. The ratio $|X(\F_q)|/g$ measures this tradeoff between the available evaluation points and the genus. We first recall the Serre bound:
	\begin{lemma}[Serre bound]\label{lm:Serre}
		Let $N = |X(\mathbb{F}_q)|$ be the number of rational points on a smooth complete curve of genus $g$ over $\mathbb{F}_q$. Then
		$$|N-(q+1)|\leq g\lfloor 2\sqrt{q} \rfloor.$$ 
	\end{lemma}
	When $q$ is a square, curves attaining the Hasse-Weil upper bound $N=q+1+2g\sqrt{q}$ are called \emph{maximal curves}.
	
	For fixed genus $g$, the Serre bound gives $|X(\F_q)|=q+O(\sqrt q)$. This allows a linear field size to provide enough sampling points for the following result.
	
	\begin{theorem}\label{thm:prob_curve}
		Fix a rational rate $R\in(0,1/2)$, then fix $\eps\in(0,1-2R)$ and a non-negative integer $g$. For all sufficiently large $n$ such that $k=Rn$ is an integer, let $q$ be a prime power satisfying
		$$q>(1+2/\eps)2^{2+32/\eps^2}n.$$
		Let $X$ be a smooth complete curve of genus $g$ defined over $\F_q$, and fix $P\in X(\F_q)$. Set $D=(k+g-1)P$ and $S=X(\F_q)\setminus\{P\}$. If $\cP$ is sampled uniformly without replacement from $S$, then $\cC_L(X,\cP,D)$ has rate $R$ and is $\eps$-relaxed HS-optimal with probability at least $1-2^{-n}$.
	\end{theorem}
	\begin{proof}
		The choice of $D$ and the Riemann-Roch theorem give
		$$\ell(D)=\deg(D)-g+1=k,$$
		so the code has rate $k/n=R$. Since $g$ is fixed, for all sufficiently large $n$ we have $\deg(D)=k+g-1\geq2g+1$, $n>\deg(D)$, $g\leq\tau$, $\ell'\leq n$, and $r\geq1$. Thus, all assumptions of \cref{thm:AG_code} hold.
		
		By \cref{lm:Serre}, $|S|\geq q-2g\sqrt q$. For all sufficiently large $n$, the assumptions on $q$ and the fixed value of $g$ imply
		$$2g\sqrt q\leq\frac q4,
		\qquad n\leq\frac q4,
		\qquad \deg(D)\leq\frac n2.$$
		Consequently, $|S|-n\geq q/2$ and
		\begin{align*}
			\frac{2\deg(D)(1/\eps_0+1)}{|S|-n}
			&\leq\frac{2n(2/\eps+1)}{q}\\
			&<2^{-1-32/\eps^2}.
		\end{align*}
		Let $F$ denote the failure-probability bound in \cref{thm:AG_code}. Since $k=Rn$ and $r\geq\eps^2n/8-1$, the preceding estimate gives
		\begin{align*}
			\log_2 F
			&\leq 2n+2k+r-r(1+32/\eps^2)\\
			&=(2+2R)n-\frac{32r}{\eps^2}\\
			&\leq(-2+2R)n+\frac{32}{\eps^2}\\
			&\leq-n
		\end{align*}
		for all sufficiently large $n$, where the last inequality uses $R<1/2$. Thus $F\leq2^{-n}$, as claimed.
	\end{proof}
    By taking $q$ to be the smallest prime power greater than the displayed lower bound, we can ensure that $q=\Theta(n)$.
	
	\begin{example}[Reed-Solomon Codes]
		Let $X=\bbP^1$ over $\F_q$, and let $P_\infty$ be its unique point at infinity. The Riemann-Roch space $L((k-1)P_\infty)$ consists of the polynomials of degree at most $k-1$. Therefore, evaluating this space at $n$ distinct points of $\bbP^1(\F_q)\setminus\{P_\infty\}=\F_q$ gives an $[n,k]_q$ Reed-Solomon code. Equivalently, when $n<q$, this code is obtained by puncturing the full-length Reed-Solomon code. Thus \cref{thm:prob_curve} shows that, for $q=\Theta(n)$, a random punctured Reed-Solomon code is $\eps$-relaxed HS-optimal with high probability, as also proved in \cite{con2024random}.
	\end{example}
	
	We now examine specific families of curves that allow for smaller field sizes.
	
	\paragraph{Sublinear Fields via Hermitian Curves} Let $q$ be a square prime power. We denote by $X_q$ the smooth projective model over $\F_q$ of the affine Hermitian curve
	$$y^{\sqrt q}+y=x^{\sqrt q+1}.$$
	Its genus is
	$$g_q=\frac{q-\sqrt q}{2},$$
	and it has $q\sqrt q+1$ rational points over $\F_q$. The curve has a unique point at infinity, denoted by $P_\infty$. A one-point Hermitian code is an AG code of the form
	$$\cC_L(X_q,\cP,mP_\infty),$$
	where the evaluation points in $\cP$ are chosen from $X_q(\F_q)\setminus\{P_\infty\}$. We show that random punctured Hermitian codes are $\eps$-relaxed HS-optimal with high probability when $q=\Theta(n^{2/3})$.

	\begin{theorem}\label{thm:prob_Hermitian}
		Fix a rational rate $R\in(0,1/2)$, then fix $\eps\in(0,1-2R)$. Set
		$$C_\eps=(1+2/\eps)2^{2+32/\eps^2}.$$
		For all sufficiently large $n$ such that $k=Rn$ is an integer, choose a square prime power $q$ satisfying
		$$(C_\eps n)^{2/3}<q\leq4(C_\eps n)^{2/3}.$$
		Let $X_q$ be the Hermitian curve above, fix $P\in X_q(\F_q)$, and set $D=(k+g_q-1)P$ and $S=X_q(\F_q)\setminus\{P\}$. If $\cP$ is sampled uniformly without replacement from $S$, then $\cC_L(X_q,\cP,D)$ has rate $R$ and is $\eps$-relaxed HS-optimal with probability at least $1-2^{-n}$.
	\end{theorem}
	\begin{proof}
		Such a $q$ always exists: take the smallest power of $4$ greater than $(C_\eps n)^{2/3}$. The upper bound on $q$ gives $g_q=(q-\sqrt q)/2=O(n^{2/3})$. Hence, for all sufficiently large $n$, we have $g_q\leq\tau$, $\deg(D)\geq2g_q+1$, $n>\deg(D)$, $\ell'\leq n$, and $r\geq1$. Thus, all assumptions of \cref{thm:AG_code} hold. The Riemann-Roch theorem also gives $\ell(D)=k$, so the code has rate $R$.
		
		Here $|S|=q\sqrt q$. Since $q^{3/2}>C_\eps n$, we have
		$$|S|-n>(C_\eps-1)n.$$
		Moreover, $g_q=o(n)$, so $\deg(D)=k+g_q-1\leq n/2$ for all sufficiently large $n$. Therefore,
		\begin{align*}
			\frac{2\deg(D)(1/\eps_0+1)}{|S|-n}
			&<\frac{(2/\eps+1)}{C_\eps-1}\\
			&<2^{-1-32/\eps^2}.
		\end{align*}
		The final inequality follows from $C_\eps=(1+2/\eps)2^{2+32/\eps^2}$. Applying the same logarithmic estimate as in the proof of \cref{thm:prob_curve} gives a failure probability at most $2^{-n}$.
	\end{proof}
	
	\paragraph{Constant Field Sizes: The Garc\'{i}a-Stichtenoth Tower}	
	For a fixed field size $q$, the asymptotic behavior of the quotient $|X(\F_q)|/g$ is limited by the Drinfeld-Vl\v{a}du\c{t} bound:
	\begin{lemma}[Drinfeld-Vl\v{a}du\c{t} bound]
		Let $N_q(g)$ denote the maximum number of rational points on a smooth projective curve of genus $g$ over $\mathbb{F}_q$. The asymptotic quantity $A(q)$ is defined and bounded as follows:
		$$A(q)=\limsup_{g\to \infty}\frac{N_q(g)}{g}\leq \sqrt{q}-1.$$
	\end{lemma}
	
	Let $p$ be a prime power and set $q=p^2$. We use the first Garc\'{i}a-Stichtenoth tower over $\F_q$, which attains this bound \cite{garcia1996asymptotic}. Its function fields are defined recursively by
	$$F_{p,1}=\F_q(x_1),\qquad
	F_{p,t+1}=F_{p,t}(x_{t+1}),$$
	where
	$$x_{t+1}^p+x_{t+1}=\frac{x_t^p}{x_t^{p-1}+1}.$$
	Let $K_{p,t}$ be the smooth complete curve over $\F_q$ with function field $F_{p,t}$. Its genus is
	$$g(K_{p,t})=
	\begin{cases}
		(p^{t/2}-1)^2 & \text{ if } t \text{ is even,}\\
		(p^{(t+1)/2}-1)(p^{(t-1)/2}-1) & \text{ if } t \text{ is odd.}\\
	\end{cases}$$
	The curve $K_{p,t}$ has one $\F_q$-rational point at infinity and at least $p^{t-1}(p^2-p)$ affine $\F_q$-rational points. Thus
	$$|K_{p,t}(\F_q)|\geq p^{t-1}(p^2-p)+1.$$
	
	\begin{theorem}\label{thm:prob_tower}
		Fix a rational rate $R\in(0,1/2)$, then fix $\eps\in(0,1-2R)$. There exists a prime power
		$$p=2^{O_R(1/\eps^2)}$$
		and an infinite family of random AG codes of rate $R$ over the fixed field $\F_{p^2}$, with lengths $n_t\to\infty$. For every sufficiently large $t$, the code of length $n_t$ is $\eps$-relaxed HS-optimal with probability at least $1-2^{-n_t}$.
	\end{theorem}
	\begin{proof}
		Write $R=a/b$ in lowest terms and set
		$$c=1+\max\left\{\frac8\eps,\frac2R,\frac2{1-R}\right\}.$$
		Choose $p$ to be the smallest power of $2$ such that $p>c+2$ and
		$$\frac{2\bigl(R(c+1)+1\bigr)(2/\eps+1)}{p-c-2}
		<2^{-1-32/\eps^2}.$$
		Since $c=O_R(1/\eps)$, this choice gives $p=2^{O_R(1/\eps^2)}$.
		For each level $t$, set
		$$n_t=b\left\lceil\frac{cp^t}{b}\right\rceil,
		\qquad k_t=Rn_t.$$
		Then $k_t$ is an integer, $k_t/n_t=R$, and $n_t\to\infty$.
		Write $g_t=g(K_{p,t})$. Let $P_\infty$ be the rational point at infinity of $K_{p,t}$, and define
		$$D_t=(k_t+g_t-1)P_\infty,
		\qquad S_t=K_{p,t}(\F_{p^2})\setminus\{P_\infty\}.$$
		
		The genus formula gives $g_t<p^t$. The choice of $c$ implies, for all sufficiently large $t$, that
		$$g_t\leq\left\lfloor\frac{\eps n_t}{8}\right\rfloor,
		\qquad \deg(D_t)\geq2g_t+1,
		\qquad n_t>\deg(D_t).$$
		Indeed, the first two inequalities follow respectively from $c>8/\eps$ and $c>2/R$, while the last follows from
		$$\limsup_{t\to\infty}\frac{\deg(D_t)}{n_t}
		\leq R+\frac1c<1.$$
		Since $2R+\eps<1$, we have $\ell'\leq n_t$ and $r\geq1$ for all sufficiently large $t$. Thus, the remaining assumptions of \cref{thm:AG_code} hold. The Riemann-Roch theorem gives $\ell(D_t)=k_t$, so each code has rate $R$.
		
		For all sufficiently large $t$, we have $n_t\leq(c+1)p^t$. Hence the rational-point bound for the tower gives
		$$|S_t|-n_t\geq p^t(p-1)-n_t
		\geq p^t(p-c-2).$$
		Moreover,
		$$\deg(D_t)<k_t+p^t
		\leq\bigl(R(c+1)+1\bigr)p^t.$$
		By the choice of $p$,
		\begin{align*}
			\frac{2\deg(D_t)(1/\eps_0+1)}{|S_t|-n_t}
			&\leq\frac{2\bigl(R(c+1)+1\bigr)(2/\eps+1)}{p-c-2}\\
			&<2^{-1-32/\eps^2}.
		\end{align*}
		If $\cP_t$ is sampled uniformly without replacement from $S_t$, applying the logarithmic estimate from the proof of \cref{thm:prob_curve} shows that $\cC_L(K_{p,t},\cP_t,D_t)$ fails to be $\eps$-relaxed HS-optimal with probability at most $2^{-n_t}$. Since $p$ is fixed, these codes form the claimed infinite family over $\F_{p^2}$.
	\end{proof}

	\appendices
	
	\crefalias{section}{appendix}
	\section{Lemma~\ref{lm:shorten}}\label{appendix:shorten}
    
    \begin{proof}
		For two different codewords $\vc_1,\vc_2\in\cC'$, $$d_I(\vc_1,\vc_2)=d_I((\vc_1,0),(\vc_2,0)).$$
		Thus 
		\begin{align*}
			d_I(\cC')
			&=\min_{\vc_1\neq\vc_2\in\cC'}d_I(\vc_1,\vc_2)\\
			&=\min_{\vc_1\neq\vc_2\in\cC'}
			d_I((\vc_1,0),(\vc_2,0))\\
			&\geq d_I(\cC).
		\end{align*}
		Suppose now that $\cC$ is HS-optimal. If $\One_{n-1}\in\cC'$, then $(\One_{n-1},0)\in\cC$, and hence
		$$
		(0,\ldots,0,1)=\One_n-(\One_{n-1},0)\in\cC,
		\qquad d_I((0,\ldots,0,1),\Zero)=2.
		$$
		This contradicts $d_I(\cC)=2(n-2k+2)\geq4$. Thus $\One_{n-1}\notin\cC'$. Moreover,
		$$d_I(\cC')\geq d_I(\cC)=2(n-2k+2)=2((n-1)-2(k-1)+1),$$
		so $\cC'$ is strict HS-optimal.
	\end{proof}
	
	\section{Lemma~\ref{lm:poly_degenerate}}\label{appendix:poly_degenerate}
	\begin{proof}
	Let $d_2$ be the dimension of the irreducible variety $Y \subseteq \mathbb{F}^{n_2}$. By applying \cref{lm:parameterize_variety} to $Y$, there exist formal power series $p_1, \dots, p_{n_2} \in \mathbb{F}[[z_1, \dots, z_{d_2}]]$ such that for any polynomial $h(\vy) \in \mathbb{F}[y_1, \dots, y_{n_2}]$, $h(\vy)$ is generically non-vanishing over $Y$ if and only if the substituted series $h(p_1, \dots, p_{n_2})$ is non-zero in $\mathbb{F}[[z_1, \dots, z_{d_2}]]$.
	
	We can naturally view $f(\vx, \vy)$ as a polynomial in $\vy$ whose coefficients are polynomials in $\vx$. By substituting the variables $y_i$ with the corresponding power series $p_i$, we obtain $f(\vx, p_1, \dots, p_{n_2})$, which is a formal power series in the variables $z_1, \dots, z_{d_2}$ whose coefficients are polynomials in $\mathbb{F}[x_1, \dots, x_{n_1}]$.
	
	Observe that each of these new polynomial coefficients is formed strictly by evaluating the $\vy$-components of $f(\vx, \vy)$. Because substituting for $\vy$ does not alter the $\vx$ variables or increase their exponents, the total degree of any resulting coefficient polynomial is bounded by the maximum degree of $\vx$ in the original polynomial $f(\vx, \vy)$. That is, the degree of any such coefficient polynomial is at most $\deg_\vx(f)$.
	
	By the premise of the lemma, $f(\vx, \vy)$ is generically non-vanishing over the product variety $X \times Y$. If every single coefficient polynomial of our constructed power series vanished completely over $X$, then for any generic point $\bar{\vx} \in X$, the evaluated power series $f(\bar{\vx}, p_1, \dots, p_{n_2})$ would be strictly zero. By \cref{lm:parameterize_variety}, this would imply that $f(\bar{\vx}, \vy)$ vanishes completely over $Y$ for all generic points $\bar{\vx}$, directly contradicting the fact that $f(\vx, \vy)$ is generically non-vanishing over $X \times Y$.
	
	Therefore, there must exist at least one specific coefficient polynomial in this series that is generically non-vanishing over $X$. We define our target polynomial $g(\vx)$ to be this exact coefficient polynomial.
	
	To conclude the proof, it remains to verify that $g(\vx)$ satisfies the three stated conditions. By its very selection, $g(\vx)$ is inherently generically non-vanishing over $X$, and as established previously, its total degree is bounded by $\deg_\vx(f)$. Finally, let $\bar{\vx} \in X$ be any point such that $g(\bar{\vx}) \neq 0$. When evaluating the formal power series at $\bar{\vx}$, the term associated with $g(\bar{\vx})$ remains strictly non-zero. This ensures that the entire evaluated series $f(\bar{\vx}, p_1, \dots, p_{n_2})$ is non-zero. By \cref{lm:parameterize_variety}, this algebraic condition guarantees that the specialized polynomial $f(\bar{\vx}, \vy)$ is generically non-vanishing over $Y$, completing the proof. 
	\end{proof}

	\section{Details of Proving Lemma~\ref{lm:code_equiv}}\label{appendix:code_equiv}
	
	To prove this lemma, we recall some algebraic geometric tools from \cite{vladut2007algebraic}. Throughout this subsection, $X$ is a smooth complete curve of genus $g$.
	
	\textbf{Degree of a map.} Note that, if $f:X \to Y$ is a non-constant rational map, then $f^*:\F(Y)\hookrightarrow \F(X)$ is a field embedding. Since the fields $\F(X)$ and $\F(Y)$ are of transcendence degree 1 over $\F$ and are finitely generated, the degree of extension $[\F(X) : f^*(\F(Y))]$ is finite; it is called the \emph{degree of the map $f$} and is denoted by $\deg f$. 
	
	\paragraph{Cartier divisors} Let $D \in\Div(X)$, $D = \sum a_P P$, and let $U$ be an open subset in $X$. By $D|_U$ we denote the following divisor on $U$:
	$$D|_U = \sum_{P \in U} a_P P.$$
	For any $D \in \Div(X)$ and any $P\in X$ there exists an open set $U=U(P,D)$ such that $P\in U$ and $D|_U$ is a principal divisor on $U$, i.e., $D|_U = (f)|_U$ for some $f \in \F(X)$. Combining the fact that a projective curve is compact, for any $D \in \Div(X)$ there exists a finite open covering $\{U_i\}$ of $X$ such that $D|_{U_i} = (f_i)|_{U_i}$. On the set $U_i \cap U_j$, the function $f_i f_j^{-1}$ is regular and has no zeroes. Conversely, for any system $(\{U_i\}, \{f_i\})$, where $\{U_i\}$ is a finite open covering of $X$, $f_i$ are rational functions, and $f_i f_j^{-1} \in \F[U_i \cap U_j]^*$, there is a unique divisor $D \in \Div X$ with $D|_{U_i} = (f_i)|_{U_i}$. A system $(\{U_i\}, \{f_i\})$ with these properties is called a \emph{Cartier divisor}.
	
	\paragraph{Functoriality} Using Cartier divisors, we can define the \emph{inverse image} $\varphi^*(D) \in \Div(X)$ of a divisor $D \in \Div(Y)$ under a regular map $\varphi:X \to Y$ of smooth irreducible complete curves. Indeed, let $(\{U_i\}, \{f_i\})$ be a Cartier divisor on $Y$, corresponding to $D \in \Div Y$. Then $\varphi^*(D)$ is defined on the covering $\{\varphi^{-1}(U_i)\}$ of $X$ by the system of functions $\{\varphi^*(f_i)\}$, where $\varphi^*(f_i)(x) = f_i(\varphi(x))$. One can easily check that $\varphi^*$ defines a group homomorphism $\varphi^*: \Div(Y) \to \Div(X)$; here, $\varphi^*(P(Y)) \subseteq P(X)$ since $\varphi^*((f)) = (\varphi^*(f))$ for any $f \in \F(Y)$. 
	
	Using Cartier divisors, for a curve $X \subset \bbP^n$ we can also define the divisor $(F)\in \Div(X)$ of a form $F$ on $\bbP^n$. Indeed, put $U_i=\{x_i\neq 0\} \cap X$ and $f_i = F/x_i^s$, where $s=\deg F$ and $(x_0:\dots:x_n)$ are homogeneous coordinates in $\bbP^n$. This system $(\{U_i\}, \{f_i\})$ defines a Cartier divisor. 
	
	There is an equivalent definition of the degree of a curve. If $L$ is a linear form, then $(L)$ is called a hyperplane section divisor. All hyperplane section divisors are linearly equivalent and therefore have the same degree, which is the degree of $X$. This agrees with the definition given above.
	
	Next we are ready to introduce a fundamental lemma in algebraic geometry, showing that the number of points in the inverse image of a point with multiplicity is the degree of the map.
	
	\begin{lemma}[Corollary 2.1.74, \cite{vladut2007algebraic}]\label{lm:pullback}
		Let $X$, $Y$ be two smooth projective curves, and let $f:X\to Y$ be a non-constant map. Then for any point $P\in Y$ we have
		$$\deg f^*(P)=\deg f,$$
		and for any divisor $D\in \Div (Y)$, we have
		$$\deg f^*(D)=(\deg D)(\deg f)$$
	\end{lemma}

	\paragraph{Embedding by linear systems} Let $D\in \Div(X)$ be a divisor on $X$ and let $k=\ell(D)$. Taking a basis $\{f_0,\dots, f_{k-1}\}$ of $L(D)$, we can define a rational map
	$$\varphi_D: X \to \bbP^{k-1} \quad P \mapsto (f_0(P):\cdots:f_{k-1}(P)).$$
	
	The map $\varphi_D$ is regular on $X\setminus \supp(D)$. For a point $P\in \supp(D)$, choose a local parameter $t_P$ at $P$, i.e. a function with $v_P(t_P)=1$. Let $m_P=\min_i \{v_P(f_i)\}$. By the definition of $L(D)$, $m_P \geq -v_P(D)$. The map $\varphi_D$ at $P$ is defined by
	$$\varphi_D(P)=\left((t_P^{-m_P}f_0)(P):\cdots:
	(t_P^{-m_P}f_{k-1})(P)\right).$$
	Since $m_P$ is the minimum valuation, at least one of the terms $(t_P^{-m_P}f_i)(P)$ is non-zero and all are finite, so the point in $\bbP^{k-1}$ is well-defined. Therefore $\varphi_D$ is a regular map from $X$ to its image $C$, which is a complete projective curve. Since $f_0,\dots,f_{k-1}$ are linearly independent, $C$ is not contained in a hyperplane $H\subset \bbP^{k-1}$.
	
	A point $P$ is called a base point of the divisor $D$ if every function $f \in L(D)$ satisfies $v_P(f)>-v_P(D)$. $D$ is said to be a base point free divisor if there is no base point. Notice that if $D$ is base point free, there exists an $i$ such that $v_P(f_i)=-v_P(D)$. Thus the parameter $m_P$ we choose is exactly $-v_P(D)$. 
	
	The Riemann-Roch theorem implies that if $\deg(D)$ is large enough, the map $\varphi_D: X\to C$ is an embedding, i.e. defines an isomorphism of $X$ onto $C$. 
	
	\begin{lemma}[Corollary 2.2.29, \cite{vladut2007algebraic}]\label{lm:embedding}
		If $\deg D \geq 2g+1$, then $\varphi_D$ is an embedding of $X$ in $\bbP^{k-1}$.
	\end{lemma}

	For $\deg(D)\geq 2g+1$, the preceding results show that $\ell(D)=\deg(D)-g+1$, that $D$ is base point free, and that $\varphi_D$ is an embedding.

	Finally we come to the proof of the lemma. It is easy to see that assuming the curve $X$, the divisor $D$ and the Riemann-Roch space are defined over $\F_q$ does not affect any of those previous propositions.
	
	\begin{proof}[Proof of \cref{lm:code_equiv}]
		Because $D=dP$ is effective, $1\in L(D)$ and we may choose the basis in the statement. By \cref{lm:embedding}, the associated map
		$$\varphi_D:X\longrightarrow\bbP^{k-1},\qquad
		Q\longmapsto[f_0(Q):\cdots:f_{k-1}(Q)]$$
		is an embedding. At $P$, let $t$ be a local parameter. Since $f_i\in L(dP)$, each $t^df_i$ is regular at $P$, and since $dP$ is base point free, these functions do not all vanish there. Hence
		$$\varphi_D(P)=[(t^df_0)(P):\cdots:(t^df_{k-1})(P)].$$
		Because $f_0=1$ and $d>0$, the first coordinate of $\varphi_D(P)$ is zero. Thus $\varphi_D(P)$ is a point at infinity.

		Let $C=\varphi_D(X)$. The pullback of a hyperplane section of $C$ is linearly equivalent to $D$. Since $\varphi_D$ has degree $1$, \cref{lm:pullback} gives $\deg(C)=\deg(D)=d$.
		The curve $C$ is irreducible and is not contained in a projective hyperplane, because $f_0,\dots,f_{k-1}$ are linearly independent.
		
		Let $U_0=\{x_0\neq0\}\subseteq\bbP^{k-1}$ and identify it with $\overline{\F}_q^{k-1}$. Since $f_0=1$, the restriction of $\varphi_D$ to $X\setminus\{P\}$ is
		$$Q\longmapsto(f_1(Q),\dots,f_{k-1}(Q)).$$
		Since $\varphi_D(P)\notin U_0$, its image is exactly $C\cap U_0$. This is an irreducible affine curve whose projective closure is $C$, so it has degree $d$. Adjoining a first coordinate equal to $1$ gives the curve
		$$Y=\{(1,f_1(Q),\dots,f_{k-1}(Q)):Q\in X\setminus\{P\}\}
		\subseteq\overline{\F}_q^k.$$
		This affine embedding preserves the degree. Moreover, since $C$ is not contained in a projective hyperplane, a generic set of $k$ points of $C$ is linearly independent. Their affine lifts in $Y$ are therefore linearly independent, so $Y$ is an MDS curve.
		
		For $\cP=(P_1,\dots,P_n)$, the $i$-th column of the generator matrix associated with the chosen basis is
		$$G_i=(1,f_1(P_i),\dots,f_{k-1}(P_i))^T=\psi(P_i)^T.$$
		Finally, $\psi$ is the restriction of an embedding and is defined over $\F_q$. Hence it maps every $S\subseteq X(\F_q)\setminus\{P\}$ bijectively onto $S'=\psi(S)\subseteq Y(\F_q)$, and this bijection preserves uniform sampling without replacement.
	\end{proof}
	
\bibliographystyle{IEEEtran}
\bibliography{ref}

\end{document}